\documentclass[conference]{IEEEtran}
\usepackage[T1]{fontenc}



\usepackage{cite}
\usepackage{amsmath,amssymb,amsthm}
\usepackage{algpseudocode}
\usepackage{graphicx}
\usepackage{mathtools}
\usepackage{xcolor}
\usepackage{algorithm}
\usepackage{bm}
\usepackage{breqn}
\usepackage{hyperref} 

\usepackage{url}
\usepackage{subfig}
\newcommand{\rmd}{\mathrm{d}}

\newtheorem{theorem}{Theorem}
\newtheorem{proposition}{Proposition}
\newtheorem{definition}{Definition}
\newtheorem{lemma}{Lemma}

\theoremstyle{remark}
\newtheorem{remark}{Remark}
\newtheorem{assumption}{Assumption}

\usepackage{authblk}
\usepackage{comment}
\IEEEoverridecommandlockouts 

\title{\LARGE \bf
A stochastic Gordon-Loeb model for optimal cybersecurity investment under clustered attacks
}

\author[1]{Giorgia Callegaro} 
\author[1]{Claudio Fontana\thanks{\underline{Corresponding author}: Claudio Fontana, Department of Mathematics ``Tullio Levi-Civita'', University of Padova, via Trieste 63, 35121 Padova, Italy. E-mail: fontana@math.unipd.it. Phone: +39-049-8271358.}}
\author[2]{Caroline Hillairet}
\author[3]{Beatrice Ongarato}

\affil[1]{Department of Mathematics ``Tullio Levi-Civita'', University of Padova, Padova, 35121, Italy.}
\affil[2]{CREST, ENSAE Paris, Palaiseau, 91120, France}
\affil[3]{Institute of Mathematical Stochastics, TU Dresden, Dresden, 01169, Germany}

\begin{document}

\maketitle
\thispagestyle{plain}
\pagestyle{plain}

\vspace{1em}

\begin{abstract}
We develop a continuous-time stochastic model for optimal cybersecurity investment under the threat of cyberattacks. The arrival of attacks is modeled using a Hawkes process, capturing the empirically relevant feature of clustering in cyberattacks. Extending the Gordon-Loeb model, each attack may result in a breach, with breach probability depending on the system's vulnerability. We aim at determining the optimal cybersecurity investment to reduce vulnerability. The problem is cast as a two-dimensional Markovian stochastic optimal control problem and solved using dynamic programming methods. Numerical results illustrate how accounting for attack clustering leads to more responsive and effective investment policies, offering significant improvements over static and Poisson-based benchmark strategies. Our findings underscore the value of incorporating realistic threat dynamics into cybersecurity risk management.
\end{abstract}

\begin{IEEEkeywords}
Cyber-risk; security breach; benefit-cost analysis; Hawkes process; stochastic optimal control.
\end{IEEEkeywords}

\section{Introduction} \label{Introduction}

Cyber-risk is nowadays widely acknowledged as one of the major sources of operational risk for organizations worldwide. The 2024 ENISA Threat Landscape Report \cite{ENISA} documents ``a notable escalation in cybersecurity attacks, setting new benchmarks in both the variety and number of incidents, as well as their consequences''. According to the AON 9th Global Risk Management Survey\footnote{\scriptsize{Source: \url{https://www.aon.com/en/insights/reports/global-risk-management-survey}.}}, cyberattacks and data breaches represent the foremost source of global risk faced by organizations, with the second biggest risk being business interruption, which is itself often a consequence of cyber-incidents.
A recent poll on Risk.net confirms information security and IT disruption as the top two sources of operational risk for 2025.\footnote{\scriptsize{Source: \url{https://www.risk.net/risk-management/7961268/top-10-operational-risks-for-2025}.}}
According to IBM, the global average cost of a data breach has reached nearly 5M USD in 2024, an increase of more than 10\%  over the previous year.\footnote{\scriptsize{Source: \url{https://www.ibm.com/reports/data-breach}.}}

The rapid and widespread emergence of cyber-risk as a key source of operational risk has led to a significant increase in cybersecurity spending. In the 2025 ICS/OT cybersecurity budget survey of the SANS Institute \cite{SANS}, 55\% of the respondents reported a substantial rise in cybersecurity budgets over the previous two years. This trend underscores the importance of adopting effective cybersecurity investment policies that balance risk mitigation with cost efficiency. 

The problem of optimal cybersecurity investment has been first addressed in the seminal work of Gordon and Loeb \cite{gordon2002economics}. In their model, reviewed in Section \ref{The Gordon-Loeb model} below, the decision maker can reduce the vulnerability to cyberattacks by investing in cybersecurity. The optimal expenditure in cybersecurity is determined by maximizing the expected net benefit of reducing the breach probability.
The Gordon-Loeb model laid the foundations for a rigorous quantitative analysis of cybersecurity investments and has been the subject of numerous extensions and generalizations: we mention here only some studies that are closely related to our context, referring to \cite{FedeleRoner2022} for a comprehensive overview.
The key ingredient of the Gordon-Loeb model is represented by the security breach probability function (see Section \ref{The Gordon-Loeb model}), which has been further analyzed in \cite{huang2013economics} and \cite{MazzoccoliNaldi2022}. The risk-neutral assumption of the original model \cite{gordon2002economics} has been relaxed to accommodate risk-averse preferences in \cite{miaoui2019enterprise}. 
The practical applicability of the Gordon-Loeb model for guiding cost-effective cybersecurity investment decisions has been examined in \cite{Gordon_et_al20} within the National Institute of Standards and Technology (NIST) cybersecurity framework.

The original Gordon-Loeb model is a static model and, therefore, does not allow to address the crucial issue of the optimal timing of investment decisions. Adopting a real-options approach, \cite{Gordon_et_al2003} and \cite{tatsumi2010optimal} have proposed dynamic versions of the model that allow analyzing the optimal timing and level of cybersecurity investment.
Closer to our setup, a dynamic extension of the Gordon-Loeb model has been developed in \cite{krutilla2021benefits}, considering the problem of optimal cybersecurity investment over an infinite time horizon and assuming that cybersecurity assets are subject to depreciation over time, while future net benefits of cybersecurity investment are discounted.

An effective cybersecurity investment policy must be adaptive and evolve in response to changing threat environments. As noted by \cite{ZellerScherer2022}, a key feature of cyber-risk is its dynamic nature, due to the rapid technological transformation and the evolution of threat actors. Similarly, \cite{BalzanoMarzi2025} emphasize the need for adaptable and responsive cybersecurity policies in order to face the challenge of dynamic cyberattacks. The framework of \cite{krutilla2021benefits} is based on a deterministic model and, therefore, cannot capture the dynamic behavior of cyber-risk.
Addressing this need, the main contribution of this work consists in proposing a modeling framework for optimal cybersecurity investment in a dynamic stochastic setup, allowing for investment policies which respond in real time to randomly occurring cyberattacks.
Our work contributes both to cyber-risk modeling and to cyber-risk management, as categorized in the recent survey by \cite{He_et_al2024}. Moreover, our stochastic modeling framework takes into account the empirically relevant feature of temporally clustered cyberattacks, as discussed in the next subsection.

\subsection{Modeling Cyberattacks with Hawkes Processes} 
\label{Hawkes processes and cyber-risk}

A distinctive feature of our modeling framework, which will be described in Section \ref{A dynamic extension}, is the use of a Hawkes process to model the arrival of cyberattacks. First introduced by Alan G. Hawkes in \cite{hawkes1971spectra}, these stochastic processes are used to model event arrivals over time and are particularly suited to situations where the occurrence of one event increases the likelihood of subsequent events (self-excitation), thereby generating temporally clustered events.
In our context, we denote by $N_t$ the cumulative number of cyberattacks up to time $t$, modeled as a Hawkes process, and by $\lambda_t$ the associated stochastic intensity (hazard rate), representing the instantaneous likelihood of an attack occurring.

The self-exciting property of the Hawkes process $(N_t)_{t\geq0}$ is captured by the specification of its stochastic intensity:
\begin{align}
\lambda_t &= \alpha +(\lambda_0-\alpha)e^{-\xi t} + \beta\sum_{i=1}^{N_t} e^{-\xi(t-\tau_i)},
\label{intensity integral form}
\end{align}
for all $t\geq0$, where
\begin{itemize}
    \item $\alpha>0$ is the long-term mean intensity;
    \item $\lambda_0>0$ is the intensity at the initial time $t = 0$;
    \item $\xi>0$ is the exponential decay rate;
    \item $\beta>0$ determines the magnitude of self-excitation;
    \item $(\tau_i)_{i\in\mathbb{N}^*}$ are the random times at which attacks occur.
\end{itemize}

Figure \ref{fig: hawkes traj} shows a simulated trajectory of $N$ and $\lambda$, showing the clustering behavior induced by the self-exciting mechanism described above. 
General presentations of the theory and the applications of Hawkes processes can be found in \cite{book_Hawkes,Lima}.

\begin{figure}[h!]
\centering
\includegraphics[width=1\linewidth]{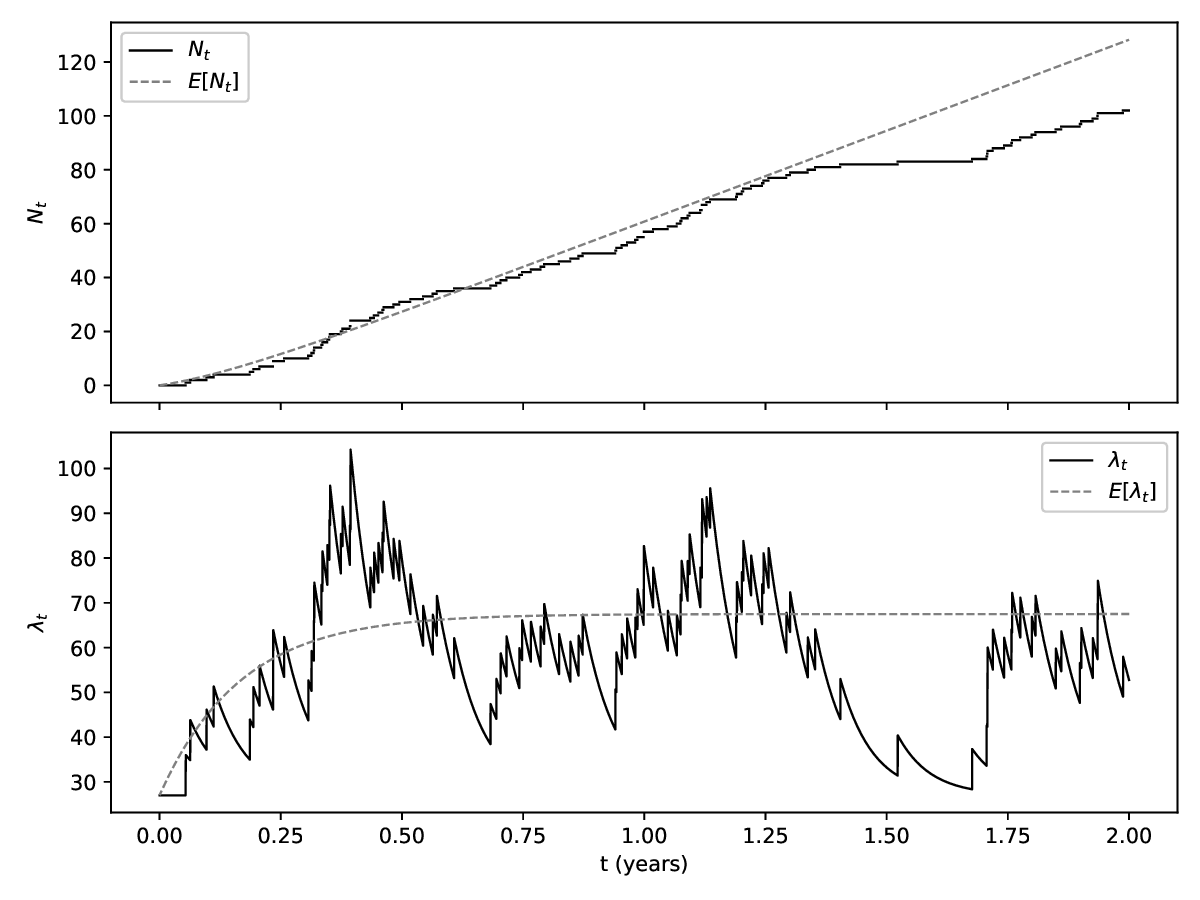}
\caption{One simulated trajectory of $N$ (top) and $\lambda$ (bottom), for $\alpha=27$, $\lambda_0=27$, $\xi=15$, $\beta=9$.}
    \label{fig: hawkes traj}
\end{figure}

This modeling choice is particularly relevant in the context of cyber-risk. Cyberattacks frequently occur in bursts, for instance following the discovery of a vulnerability or due to the propagation of malware across interconnected systems (see \cite{Nguyen_et_al2024}). Such clustered patterns are not adequately captured by memoryless models, such as those based on Poisson processes.

Empirical evidence supports the appropriateness of the Hawkes framework for modeling cyber-risk. A contagious behavior in cyberattacks has been documented in \cite{baldwin2017contagion}, analyzing the threats to key internet services using data from the SANS Institute. Using data from from the Privacy Rights Clearinghouse, it has been empirically demonstrated in \cite{bessy2021multivariate} that Hawkes-based models provide a more realistic representation of the interdependence of data breaches compared to Poisson-based models. The recent work \cite{boumezoued2023cyber} reinforces this perspective by calibrating a two-phase Hawkes model to cyberattack data taking into
account publication of cyber-vulnerabilities.
These studies provide strong support for modeling cyberattacks via Hawkes processes, as described in more detail in Section \ref{A dynamic extension}.

\subsection{Goals and Main Contributions}
In this work, we address the challenge of optimal cybersecurity investment under temporally clustered cyberattacks, in line with the empirical evidence reported above. In particular, we aim at studying the adaptive investment policy that best responds in real time to the random arrival of cyberattacks, within a framework that balances realism with analytical tractability.
To this end, we develop a continuous-time stochastic extension of the classical Gordon-Loeb model, describing attack arrivals with a Hawkes process. The model incorporates key operational features such as technological obsolescence and the decreasing marginal effectiveness of large investments. The resulting optimization problem is cast as a stochastic optimal control problem and solved via dynamic programming methods. We develop efficient numerical schemes to compute the optimal policy and we quantify the benefits of dynamic investment strategies under clustered attacks.
By integrating risk dynamics into the cybersecurity investment problem, our framework provides new insights into how organizations can better allocate resources to mitigate cyber-risk.

\subsection{Structure of the Paper}
In Section \ref{Model}, we recall the original Gordon-Loeb model and introduce our continuous-time stochastic extension. In Section \ref{The optimization problem}, we formulate the cybersecurity investment problem and characterize the optimal policy. Section \ref{Methods} details the model parameters and the numerical methods used in our analysis. Section \ref{Numerical results} presents the results of our numerical analysis and discusses their practical implications for cyber-risk management. 
Section \ref{InsurancePart} discusses the relevance of our study from an insurance perspective, also in relation to the calculation of premia for cyber-insurance contracts.
Section \ref{Conclusions} concludes.

\section{The Model} \label{Model}

We study the decision problem faced by an entity (a public administration or a large corporation) that is threatened by a massive number of randomly occurring cyberattacks with a temporally clustered pattern. As in the Gordon-Loeb model (reviewed in Section \ref{The Gordon-Loeb model}), not all cyberattacks result in successful breaches of the entity's  system. The success rate of each attack depends on the system's vulnerability, which the entity can mitigate by investing in cybersecurity.

\subsection{The Gordon-Loeb Model} \label{The Gordon-Loeb model}

The Gordon-Loeb model, introduced in 2002 in the seminal work \cite{gordon2002economics}, provides a static framework for determining the optimal investment in cybersecurity to protect a given information set under the threat of cyberattacks. In our context, the information set corresponds to the entity's IT infrastructure.
In the Gordon-Loeb model, the information set is characterized by three key parameters, all assumed to be constant:
\begin{itemize}
\item $p \ {\in [0,1]}$: the probability that a cyberattack occurs;
\item $v \ {\in [0,1]}$: the baseline probability that a cyberattack successfully breaches the information set (vulnerability);
\item $\ell \ {\ge 0}$: the loss incurred when a breach occurs.
\end{itemize}

The vulnerability $v$ refers to a baseline level of security, reflecting the cybersecurity measures already in place before any new decision is taken. Without any further investment in cybersecurity, the expected loss is $vp\ell$. 
To mitigate its vulnerability, the entity may invest an additional amount $z \geq 0$ in cybersecurity. The variable $z$ therefore represents incremental cybersecurity investment beyond the existing baseline protection. In particular, $z = 0$ means that no further investment is made, but the system still retains the prior protection embodied in $v$.
Throughout the paper, absence of cybersecurity investment will refer to the case $z = 0$.
The effectiveness of the investment $z$ is measured by a security breach probability function $S(z, v)$, which represents the probability that an attack successfully breaches the information set, given investment level $z$ and the vulnerability $v$. The resulting expected loss is thus $S(z, v)p\ell$.
Gordon and Loeb require the function $S$ to satisfy the properties listed in the following assumption.

\begin{assumption} \label{assA}
$ $
\begin{itemize}
    \item[(A1)] $S(z, 0) = 0$, for all $z\geq0$, i.e., an invulnerable information set always remains invulnerable;
    \item[(A2)] $S(0, v) = v$, i.e., in the absence of  investment, the information set retains its baseline vulnerability; 
    \item[(A3)]
    $S$ is decreasing and convex in $z$, so that $S_z(z, v) < 0$ and $S_{zz}(z, v) > 0$, for all $z\geq0$, i.e., cybersecurity investment reduces breach probability with diminishing marginal effectiveness.
\end{itemize}    
\end{assumption}

Gordon and Loeb consider two classes of security breach probability functions, which satisfy Assumption (A): 
\begin{equation}    \label{class of S}
S_I(z, v) = \frac{v}{(az+1)^b}
\qquad\text{and}\qquad
S_{II}(z, v) = v^{a z+1},
\end{equation}
for some parameters $a, b>0$.

In \cite{gordon2002economics}, the optimal investment in cybersecurity is determined by maximizing the Expected Net Benefit of Investment in information Security (ENBIS), defined as follows: 
\begin{equation}\label{ENBIS}
\text{ENBIS}(z) := \bigl(v - S(z, v)\bigr)p \ell - z.
\end{equation}
The ENBIS function quantifies the net trade-off between the benefit (captured by the reduction in the expected loss due to the investment in cybersecurity) and the direct cost of investing.
Under Assumption \ref{assA}, the optimal investment level $z^*$ is determined by the following first-order condition:
\[
-S_z(z^*,v)p \ell -1=0.
\]

\begin{remark}
For both classes of security breach functions in \eqref{class of S}, Gordon and Loeb show that the optimal cybersecurity investment never exceeds $1/e\approx 37\%$ of the expected loss:
\[
z^*< \frac{v p \ell}{e}.
\]    
\end{remark}

\subsection{A Continuous-Time Model Driven by a Hawkes Process} \label{A dynamic extension}

We now introduce a continuous-time model for randomly occurring  cyberattacks. As discussed in Section \ref{Hawkes processes and cyber-risk}, we want to capture the empirically relevant feature of clustering of cyberattacks, while retaining the key elements of the original Gordon-Loeb model reviewed in Section \ref{The Gordon-Loeb model}.

The arrival of cyberattacks is described by a Hawkes process $N=(N_t)_{t\geq0}$, defined on a probability space $(\Omega, \mathcal{F}, \mathbb{P})$, with $N_t$ representing the number of cyberattacks up to time $t$, for all $t\geq0$. As explained in Section \ref{Hawkes processes and cyber-risk}, the process $N$ is characterized by a self-exciting stochastic intensity $(\lambda_t)_{t\geq0}$ solving the following stochastic differential equation:
\begin{equation}    \label{intensity-SDE}
\rmd \lambda_t= \xi(\alpha-\lambda_t)\, \rmd t + \beta\rmd {N}_t, \quad \lambda_0>0,
\end{equation}
whose explicit solution is given by \eqref{intensity integral form}.
In the next proposition (adapted from \cite{dassios2013exact}), we compute the expectation of some basic quantities which will be used later.\footnote{The condition $\beta < \xi$ in Proposition \ref{expectation lambda N} ensures that the $L^1$-norm of the self-excitation kernel of the Hawkes process is strictly less than one. This guarantees that the process is non-explosive, meaning that it generates almost surely a finite number of events over any finite time interval. The same condition also corresponds to the stationarity condition, widely adopted in the theory of Hawkes processes since the seminal work of \cite{HawkesOakes74}.}

\begin{proposition} \label{expectation lambda N}
Let $(N_t)_{t\geq0}$ be a Hawkes process with intensity $(\lambda_t)_{t\geq0}$ given in \eqref{intensity-SDE}, with $\beta<\xi$. Then, for all $t\geq0$,
   \begin{align*}
       \mathbb{E}[\lambda_t]&=\frac{\alpha\xi}{\xi-\beta} + e^{-(\xi-\beta )t}\left( \lambda_0-\frac{\alpha\xi}{\xi-\beta} \right), \\
       \mathbb{E}[N_t]&=\int_0^t \mathbb{E}[\lambda_s] \rmd s \\
       &= \frac{\alpha\xi}{\xi-\beta} t - \frac{1}{\xi-\beta} \left(\lambda_0-\frac{\alpha\xi}{\xi-\beta} \right) \left(e^{-(\xi-\beta)t}-1\right).
   \end{align*} 
\end{proposition}

We denote by $(\tau_i)_{i\in\mathbb{N}^*}$ the random jump times of the process $N$, representing the arrival times of cyberattacks. 
For each $t\geq0$, we denote by $\mathcal{F}_t:=\sigma(N_s;s\leq t)$ the sigma-field generated by the Hawkes process $N$ up to time $t$, representing the information generated by the history of the attack timings up to time $t$. In the following, this is assumed to be the information set available to the decision maker (see also Remark \ref{rem:info} below).

In the absence of cybersecurity investment, each attack is assumed to breach the entity's IT system with fixed probability $v$ (vulnerability). 
For each $i\in\mathbb{N}^*$, we introduce a Bernoulli random variable $B^v_i$ of parameter $v$, where the event $\{ B^v_i = 1 \}$ corresponds to a successful breach caused by the $i$-th attack. In the event of a breach, the entity incurs a random monetary loss $\eta_i$, realized at the attack time $\tau_i$. Otherwise, if $B^v_i = 0$, the attack is blocked and no loss occurs at time $\tau_i$.

The families of random variables $(B^v_i)_{i\in\mathbb{N}^*}$ and $(\eta_i)_{i\in\mathbb{N}^*}$ are assumed to satisfy the following standing assumption.

\begin{assumption}  \label{assB}
The family $(\eta_i)_{i\in\mathbb{N}^*}$ is composed by i.i.d. positive random variables in $L^1(\mathbb{P})$.
The families $(\eta_i)_{i\in\mathbb{N}^*}$ and $(B^v_i)_{i\in\mathbb{N}^*}$ are mutually independent and independent of $N$.
\end{assumption}

\begin{remark}
The assumption that the losses $(\eta_i)_{i\in\mathbb{N}^*}$ are independent of the attack arrival process $N$ can be restrictive from an empirical viewpoint. In practice, both the frequency and the severity of cyber losses can increase during periods of clustered attacks.
From a modeling standpoint, this assumption can be relaxed, for instance by allowing the loss distribution to depend on the current intensity, making use of the  techniques recently developed in \cite{MSPD}.
However, introducing such dependency would significantly complicate the analysis of the optimization problem  in Section \ref{The optimization problem} and  requires precise information on the conditional distribution of losses given the attack dynamics. In the absence of sufficient publicly available data on cyber losses, we retain the independence assumption as a tractable and transparent modeling choice.
\end{remark}

The cumulative loss incurred over a planning horizon $[0,T]$, in the absence of cybersecurity investment, is given by:
\begin{equation}\label{actual losses without security}
L_T^0 := \sum_{i=1}^{N_T} \eta_i B_i^v.
\end{equation}

In our dynamic model, the entity can react to the evolving threat environment by investing in cybersecurity, in order to mitigate its vulnerability to cyberattacks. Investment occurs continuously throughout the planning horizon $[0,T]$ and is described by a non-negative investment rate process $(z_t)_{t \in [0, T]}$. For each $t<T$, the quantity $z_t$ represents the increase in the level of cybersecurity over the infinitesimal time interval $[t,t+\rmd t]$ and is assumed to be chosen based on the information set $\mathcal{F}_{t-}$ available to the decision maker just before time $t$.

\begin{remark}  \label{rem:info}
In the terminology of stochastic processes, the requirement that the investment rate $z_t$ must be chosen on the basis of $\mathcal{F}_{t-}$ corresponds to requiring the process $(z_t)_{t\in[0,T]}$ to be predictable with respect to the filtration generated by $N$. In particular, at any attack time $\tau_i$, the value $z_{\tau_i}$ is determined by the information $\mathcal F_{\tau_i-}$: the decision maker may adjust investment after observing an attack, but not in a way that anticipates the attack at the same instant.
We point out that, in our setup, the outcomes of previous attacks (i.e., whether an attack has successfully breached the system or not) do not carry any relevant informational content for decision making, as they do not affect the dynamics of future attack arrivals.
\end{remark}

Investment in cybersecurity is subject to rapid technological obsolescence (see, e.g., \cite{HayesBodhani}). In line with \cite{krutilla2021benefits}, we take into account this significant aspect in our model by introducing a depreciation rate $\rho>0$. The cybersecurity level reached at time $t$ is then defined as follows, for all $t\in[0,T]$:
\begin{equation}\label{h eq integral}
    H^z_t=H_0e^{-\rho t}+\int_0^t e^{-\rho(t-s)}z_s \rmd s,
\end{equation}
which equivalently, in differential form, reads as follows: 
\[
\rmd H^z_t= (z_t-\rho H^z_t)\rmd t,
\qquad H^z_0=H_0\geq0.
\]
As in \cite{krutilla2021benefits}, we interpret the cybersecurity level as an aggregated asset, which can be thought of as a combination of technological infrastructures, software, and human expertise.

In our continuous-time framework, we let the breach probability evolve dynamically with the current cybersecurity level. More specifically, suppose that the decision maker adopts an investment policy $z=(z_t)_{t\in[0,T]}$. At each attack time $\tau_i$, a breach is assumed to occur with probability
\begin{equation}    \label{SHtaui}
S(H^z_{\tau_i},v),
\end{equation}
where $H^z_{\tau_i}$ is given by \eqref{h eq integral} evaluated at $t=\tau_i$ and $S$ is a security breach probability function satisfying Assumption \ref{assA}, as in the original Gordon-Loeb model. Hence, the probability that the $i$-th attack successfully breaches the IT system depends on the cybersecurity level $H^z_{\tau_i}$ reached at the attack's time $\tau_i$. In turn, $H^z_{\tau_i}$ is determined by the investment realized over the time period $[0,\tau_i]$, taking into account technological obsolescence. If the $i$-th attack breaches the IT system, then the entity incurs into a loss of $\eta_i$, otherwise the attack is blocked and the entity does not suffer any loss at time $\tau_i$.

\begin{remark}
The proposed model allows for adaptive real-time cybersecurity investment. More specifically, the arrival of an attack triggers an increased likelihood of further attacks within a short timeframe, due to the form \eqref{intensity-SDE} of the intensity. The decision maker can respond in real-time by increasing cybersecurity investment, which in turn reduces future breach probabilities through the function $S$ in \eqref{SHtaui}.
The optimal investment policy will be determined in Section \ref{The optimization problem}, while the practical importance of allowing for an adaptive real-time investment strategy - rather than a static policy as in the original Gordon-Loeb model - will be empirically analyzed in Section \ref{Comparison with original Gordon-Loeb}.
\end{remark}

Analogously to the case without investment in cybersecurity, we can write as follows the cumulative losses $L^z_T$ incurred on the time interval $[0,T]$ when investing in cybersecurity according to a generic rate  $z=(z_t)_{t\in[0,T]}$:
\begin{equation} \label{actual losses with security}
L_T^z := \sum_{i=1}^{N_T} \eta_i B_i^{S(H^z_{\tau_i},v)},
\end{equation}
where $(B_i^{S(H^z_{\tau_i},v)})_{i\in\mathbb{N}^*}$ is a family of random variables taking values in $\{0,1\}$ and satisfying the following assumption.

\begin{assumption}  \label{assC}
For any process $(z_t)_{t\in[0,T]}$, it holds that
\[
\mathbb{P}\Bigl(B_i^{S(H^z_{\tau_i},v)}=1 \Big| \mathcal{F}_T\Bigr) = S(H^z_{\tau_i},v),
\qquad\text{ for all } i\in\mathbb{N}^*,
\]
where $(H^z_t)_{t\in[0,T]}$ is determined by $(z_t)_{t\in[0,T]}$ as in \eqref{h eq integral}. Moreover, for each $i\in\mathbb{N}^*$, the random variables $B_i^{S(H^z_{\tau_i},v)}$ and $\eta_i$ are conditionally independent given  $\mathcal{F}_T$.
\end{assumption}

\begin{remark}\label{point_process}
The cumulative loss process $(L^z_t)_{t\in[0,T]}$ defined as in \eqref{actual losses with security} constitutes a {\em marked Hawkes process}, in the terminology of point processes (see \cite{bremaud}). In our modeling framework, the marks (losses) are endogenous and depend on the dynamically evolving cybersecurity level $(H^z_t)_{t\in[0,T]}$.
\end{remark}

For strategic decision making, a key quantity is represented by the expected losses due to cyberattacks over the time interval $[0,T]$ when adopting a suitable cybersecurity policy. This is the content of the following proposition, which will be fundamental for addressing the optimal investment problem in Section \ref{The optimization problem}. 
We denote by $\bar{\eta}:=\mathbb{E}[\eta_i]$ the expected loss resulting from a successful breach, for all $i\in\mathbb{N}^*$.

\begin{proposition} \label{prop:expectations}
Under Assumptions \ref{assB} and \ref{assC}, it holds that
    \begin{align*}
     \mathbb{E}[L_T^0]&=  \bar{\eta}\,v\,\mathbb{E}\left[\int_0^T  \lambda_t \rmd t\right], \\ 
          \mathbb{E}[L_T^z]&=  \bar{\eta}\,\mathbb{E}\left[\int_0^T  S(H^z_t, v)\lambda_t \rmd t\right].
     \end{align*}
Therefore, the expected net benefit of investment is 
     \begin{align}\mathbb{E}[L_T^0-L_T^z]&= \bar{\eta}\,\mathbb{E}\left[\int_0^T \bigl(v-S(H^z_t, v)\bigr) \lambda_t \rmd t \right]. \label{diff exp}
    \end{align}
\end{proposition}
\begin{proof}
Let $z=(z_t)_{t\in[0,T]}$ be an arbitrary cybersecurity investment rate process. Applying the tower property of conditional expectation and making use of Assumptions \ref{assB} and \ref{assC}, we can compute
\begin{align*}
\mathbb{E}[L^z_T]
&= \mathbb{E}\left[\sum_{i=1}^{N_T} \eta_i B_i^{S(H^z_{\tau_i},v)}\right] \\
&= \mathbb{E}\left[\sum_{i=1}^{N_T}\mathbb{E}\left[ \eta_i  B_i^{S(H^z_{\tau_i},v)}\Big|\mathcal{F}_T\right]\right]    \\
&= \bar{\eta}\,\mathbb{E}\left[\sum_{i=1}^{N_T}S(H^z_{\tau_i},v)\right]   \\
&= \bar{\eta}\,\mathbb{E}\left[\int_0^TS(H^z_t,v)\rmd N_t\right]   \\
&= \bar{\eta}\,\mathbb{E}\left[\int_0^TS(H^z_t,v)\lambda_t\rmd t\right],
\end{align*}
where the last step follows by definition of intensity (see, e.g., \cite[Definition II.D7]{bremaud}), together with the continuity of the process $H^z$, which implies that $H^z_{\tau_i}$ coincides with the left-limit $H^z_{\tau_i-}$, for all $i\in\mathbb{N}^*$. 
The first equation in the statement of the proposition follows as a special case by taking $z\equiv 0$.
\end{proof}

\section{Optimal Cybersecurity Investment} \label{The optimization problem}

In this section, we determine the optimal cybersecurity investment policy, in the model setup introduced in Section \ref{A dynamic extension}. In the spirit of the original Gordon-Loeb model, we aim at characterizing the investment rate process $z^*=(z^*_t)_{t\in[0,T]}$ which maximizes the net trade-off between the benefits and the costs of cybersecurity over a planning horizon $[0,T]$.

To ensure the well-posedness of the optimization problem, we constrain the admissible investment policies to a suitably defined admissible set $\mathcal{Z}$.\footnote{We point out that, as a direct consequence of the Cauchy-Schwarz inequality, the integral in \eqref{h eq integral} is always well-defined for every process $z\in\mathcal{Z}$.}
In line with Remark \ref{rem:info}, in the following definition we restrict our attention to processes that are predictable with respect to the filtration generated by $N$.

\begin{definition}  
\label{def:adm}
The admissible set $\mathcal{Z}$ is defined as the set of all non-negative predictable processes $(z_t)_{t\in[0,T]}$ such that $\mathbb{E}[\int_0^T z_t^2 \rmd t]< \infty$.
\end{definition}

We now formulate the central optimization problem, which generalizes the benefit-cost trade-off function in \eqref{ENBIS} to a dynamic stochastic setting. The objective is to maximize the expected net benefit of cybersecurity investments:
\begin{equation} \label{linear criterion u z}
\sup_{z \in \mathcal{Z}} \mathbb{E}\left[  L_T^0-L_T^z- \int_0^T \left(\delta z_t+\frac{\gamma}{2}z_t^2\right) \rmd t + U(H^z_T)\right ],
\end{equation}
where $L^0_T$ and $L^z_T$ are defined in \eqref{actual losses without security} and \eqref{actual losses with security}, respectively, and the state variables $\lambda$ and $H^z$ satisfy the dynamics
\begin{align}
\rmd \lambda_t &= \xi(\alpha-\lambda_t)\rmd t + \beta\rmd {N}_t, 
\label{sde lambda} \\
\rmd H^z_t &= (z_t-\rho H^z_t)\rmd t.
\label{sde h}
\end{align}
In the objective functional \eqref{linear criterion u z}, the term $\mathbb{E}[L_T^0-L_T^z]$ represents the reduction in the expected losses due to the investment in cybersecurity. Differently from \eqref{ENBIS}, we consider in \eqref{linear criterion u z} a non-linear cost function $z\mapsto \delta z+\gamma z^2/2$, for $\delta,\gamma>0$. 
The non-linearity penalizes irregular or highly concentrated investment strategies, reflecting real-world constraints and incentivizing smoother, more sustained cybersecurity efforts (e.g., continuous IT updates versus abrupt large-scale interventions).
The term $U(H^z_T)$ accounts for the residual utility of the cybersecurity level reached at the end of the planning horizon. This accounts for the fact that cybersecurity investment carries a long-term benefit, since the entity does not cease to exist after the planning horizon. As usual, the function $U:\mathbb{R}_+\to\mathbb{R}$ is assumed to be an increasing and concave utility function.

Up to a rescaling of the model parameters, there is no loss of generality in taking $\delta=1$. Hence, making use of Proposition \ref{prop:expectations}, we can equivalently rewrite problem \eqref{linear criterion u z} as follows:
\begin{equation} \label{linear criterion new z}
\sup_{z \in \mathcal{Z}} \mathbb{E}\left[\int_0^T \left( \bar{\eta}\bigl(v-S(H^z_t,v)\bigr)\lambda_t-z_t-\frac{ \gamma}{2}z_t^2 \right) \rmd t +  U(H^z_T)\right ].
\end{equation}

\begin{remark}
The objective functional \eqref{linear criterion u z} is linear in the losses $(\eta_i)_{i\in\mathbb{N}^*}$. More specifically, formula \eqref{linear criterion new z} shows that the loss distribution enters into the optimization problem only through its first moment.
In turn, this implies that the optimal policy derived in Theorem \ref{thm:hjb} below depends only on $\bar{\eta}$ and not on higher-order moments of the loss distribution. Under nonlinear objective functionals, higher moments and the tail behavior of the loss distribution would affect the optimal policy.
\end{remark}

Problem \eqref{linear criterion new z} is a bi-dimensional stochastic optimal control problem, where the stochastic intensity process $(\lambda_t)_{t\in[0,T]}$ acts as an additional state variable beyond the controlled process $(H^z_t)_{t\in[0,T]}$. 
Due to the Markovian structure of the system, dynamic programming techniques can be applied for the solution of \eqref{linear criterion new z}.
To this end, we first introduce the following notation, for any $(t,\lambda,h)\in[0,T]\times(0,\infty)\times\mathbb{R}_+$:
\begin{itemize}
\item for all $s\in[t,T]$,
\begin{equation}\label{h markovian}
H_s^{t,h;z} := he^{-\rho(s-t)} + \int_t^s e^{-\rho(s-v)} z_v \rmd v
\end{equation}
 represents the cybersecurity level reached at time $s$ when starting from level $H_t=h$ at time $t$ and investing according to $z\in\mathcal{Z}$;
\item for all $s\in[t,T]$,
\[
\lambda_s^{t, \lambda} := \alpha +(\lambda-\alpha)e^{-\xi (s-t)} + \beta\sum_{i=N_t+1}^{ N_s} e^{-\xi(s-\tau_i)}
\]
 represents the stochastic intensity at time $s$ when starting from value $\lambda_t=\lambda$ at time $t$.
\end{itemize}
For any stopping time $\tau$ taking values in $[0,T]$, we denote by $\mathcal{Z}_{\tau}$ the admissible set $\mathcal{Z}$ restricted to the stochastic time interval $[\tau,T]$. 

We define as follows the benefit-cost trade-off functional $J$ associated to a given investment rate process $z$:
\begin{align*}
J(t,\lambda, h;z) &:= \mathbb{E}\left[\int_t^T \bar{\eta}\bigl(v-S(H_s^{t,h;z},v)\bigr)\lambda_s^{t,\lambda}\rmd s \right.\\
&\qquad\quad\left.
- \int_t^T\left(z_s+\frac{ \gamma}{2}z_s^2\right) \rmd s+ U(H_T^{t,h;z})\right].
\end{align*}
As a consequence of Definition \ref{def:adm} together with the concavity of the function $U$, the functional $J$ is always well-defined and finite.
Consequently, the value function associated to the stochastic optimal control problem \eqref{linear criterion new z} is given by
\begin{equation} \label{value function}
V(t,\lambda, h) := \sup_{z\in \mathcal{Z}_t} J(t,\lambda, h;z).
\end{equation}
In our dynamic model, the value function $V(t,\lambda,h)$ encodes the benefit-cost trade-off of cybersecurity investment over the residual planning horizon $[t,T]$, when considered at time $t$ with current cybersecurity level $h$ and intensity $\lambda$.

\begin{remark}\label{rem:propertiesV}
Standard arguments allow to prove that the function $V$ has the following behavior with respect to $\lambda$ and $h$:
\begin{itemize}
\item for every $(t,h)\in[0,T)\times\mathbb{R}_+$, the map $\lambda\mapsto V(t,\lambda,h)$ is strictly increasing: 
the benefit of cybersecurity investment is greater in the presence of a greater risk of cyberattacks;
\item for every $(t,\lambda)\in[0,T]\times(0,\infty)$, the map $h\mapsto V(t,\lambda,h)$ is strictly increasing and concave: increasing the current cybersecurity level $h$ always improves the expected future benefit, but the marginal value of additional protection decreases as the cybersecurity level $h$ grows. 
\end{itemize}
\end{remark}

The optimization problem formulated in \eqref{linear criterion u z} falls within the framework of the optimal control of a piecewise deterministic Markov process, in the sense of \cite{DavisPDMP}. Under the admissibility conditions imposed in Definition \ref{def:adm}, the performance functional $J$ is well-posed for every control $z\in\mathcal{Z}_t$. As a consequence, the same techniques used in the proof of \cite[Theorem 3.3.1]{Pham} yield that the value function satisfies the dynamic programming principle. 
More precisely, for all $(t,\lambda,h)$ in $[0,T) \times (0,\infty) \times \mathbb{R}_+$ and for every stopping time $\tau$ taking values in $[t,T]$, it holds that
\begin{align*}
V(t,\lambda, h) &= \sup_{z\in \mathcal{Z}_t} \mathbb{E}\left[
\int_t^{\tau}\bar{\eta}\bigl(v-S(H_s^{t,h;z},v) \bigr)\lambda_{s}^{t,\lambda}\rmd s \right.\\
&\qquad\quad\left. -\int_t^{\tau}\left(z_s+\frac { \gamma}{2} z_s^2 \right) \rmd s + V(\tau, \lambda_{\tau}^{t,\lambda}, H_{\tau}^{t,h;z})\right].
\end{align*}

We now proceed to characterize the value function $V$ as the solution to a Hamilton-Jacobi-Bellman (HJB) partial integro-differential equation (PIDE). Further theoretical results on the value function, together with a verification theorem under suitable regularity assumptions, are proved in \cite[Section 2.3]{ongarato2026thesis}.

\begin{theorem}\label{thm:hjb}
Suppose that the value function $V$ defined in \eqref{value function} is of class $\mathcal{C}^{1,1,1}$ (i.e., continuously differentiable in all its arguments). Then, the function $V$ solves the following HJB-PIDE:\footnote{For brevity of notation, in the statement and in the proof of this theorem, we omit to denote explicitly the dependence of $V$ on its arguments $(t,\lambda,h)$.}
\begin{equation}\label{hjb obsolescence}
\begin{aligned}
& \frac{\partial V}{\partial t} +\xi(\alpha-\lambda)\frac{\partial V}{\partial \lambda}-\rho h \frac{\partial V}{\partial h} \\
&\; + \lambda(V(t,\lambda+\beta,h)-V(t,\lambda,h)) + \bar{\eta}(v-S(h,v))\lambda \\
&\; + \frac{\bigl(\left(\frac{\partial V}{\partial h}-1\right)^+\bigr)^2}{2\gamma} =0,\\ 
& V(T,\lambda, h)= U(h).
\end{aligned}
\end{equation}
In addition, the optimal investment rate process $z^*$ is given by 
\begin{equation}\label{optimal control}z^*=\frac{\left(\frac{\partial V}{\partial h}-1\right)^+}{\gamma}.\end{equation}
\end{theorem}
\begin{proof}
In view of the assumption that $V$ is of class $\mathcal{C}^{1,1,1}$, standard arguments based on It\^o's formula together with \eqref{sde lambda} and \eqref{sde h} imply that $V$ satisfies the following HJB equation (see, e.g., \cite[Section 5.2]{bensoussan2024stochastic}):
\begin{align*}
0 &= \sup_{z\geq0} \left(\frac{\partial V}{\partial t} +\xi(\alpha-\lambda)\frac{\partial V}{\partial \lambda} -\rho h \frac{\partial V}{\partial h} + z  \frac{\partial V}{\partial h}\right.\\
&\qquad\qquad\left.
+\lambda\bigl(V(t,\lambda+\beta,h)-V(t,\lambda,h)\bigr) \right.\\
&\qquad\qquad\left. 
+ \bar{\eta}(v-S(h,v))\lambda - z- \frac{\gamma}{2}z^2 \right) \\
&= \frac{\partial V}{\partial t} +\xi(\alpha-\lambda)\frac{\partial V}{\partial \lambda} -\rho h \frac{\partial V}{\partial h}\\
&\quad
+ \lambda\bigl(V(t,\lambda+\beta,h)-V(t,\lambda,h)\bigr) + \bar{\eta}(v-S(h,v))\lambda \\
&\quad + \sup_{z\geq0} \left(z \frac{\partial V}{\partial h}- z- \frac{\gamma}{2}z^2\right).
\end{align*}
The supremum in the last line is given by
\[
\sup_{z\geq 0} \left(z \frac{\partial V}{\partial h}- z -\frac{\gamma}{2}z^2\right)= \begin{cases}
0, & \text{if }\frac{\partial V}{\partial h}\leq  1, \\
    \frac{1}{2 \gamma}\left(\frac{\partial V}{\partial h}-1 \right)^2, & \text{otherwise,}
\end{cases} 
\]
and is reached by the optimal control given in equation \eqref{optimal control}.
\end{proof}

 A numerical method for the solution of the PIDE \eqref{hjb obsolescence} will be presented in Section \ref{numerics PIDE} and then applied in Section \ref{Numerical results}.

The optimal investment rate $z^*_t$ in equation \eqref{optimal control} depends on current time $t$, on the current level $\lambda_t$ of the stochastic intensity and on the current cybersecurity level $H_t$. In particular, the dependence on $\lambda_t$ makes $z^*_t$ adaptive, meaning that it reacts to the random arrival of cyberattacks. Since the arrival of an attack increases the likelihood of further attacks, due to the self-exciting behavior of the Hawkes process, this enables the decision maker to strategically increase the cybersecurity investment in order to raise the cybersecurity level as a defense for the incoming attacks. This important feature will be numerically illustrated in Section \ref{Optimal control along a trajectory}.

\begin{remark}
The optimal policy described in equation \eqref{optimal control} admits a clear economic interpretation: it is worth investing in cybersecurity whenever the marginal benefit of the investment is greater than its marginal cost. This insight aligns with the earlier results of \cite{krutilla2021benefits} in a dynamic but deterministic setup.
\end{remark}

We close this section with the following result, which provides an explicit lower bound for the value function $V$. 

\begin{proposition} \label{lower bound}
For every $(t,\lambda,h)\in[0,T]\times(0,\infty)\times\mathbb{R}_+$, it holds that
\begin{equation}\label{Vbound}
V(t,\lambda,h)\geq J(t, \lambda, h; \rho h),
\end{equation}
where 
\begin{align*}
J(t, \lambda, h; \rho h) &= U(h)
- \rho h\left(1+\frac{\gamma}{2}\rho h\right)(T-t) \\
&\quad + \bar{\eta}\bigl(v-S(h,v)\bigr) 
    \left(\frac{\alpha\xi}{\xi-\beta} (T-t) \right.\\
&\quad\left.
- \frac{1}{\xi-\beta} \Big(\lambda-\frac{\alpha\xi}{\xi-\beta}\Big) \left(e^{-(\xi-\beta)(T-t)}-1 \right)\right).
\end{align*}
\end{proposition}
\begin{proof}
By definition of the value function \eqref{value function}, it holds that $V(t,\lambda,h)\geq J(t,\lambda,h,z)$, for any given $z\in\mathcal{Z}_t$. In particular, the constant process $\bar{z}\equiv \rho h$ belongs to $\mathcal{Z}_t$ and, therefore, inequality \eqref{Vbound} holds.
In view of equation \eqref{h markovian}, we have that
\[
H_s^{t,h;\rho h}
= he^{-\rho(s-t)} + \int_t^s e^{-\rho(s-v)} \rho h \rmd v
= h,
\]
for all $s\in[t,T]$. Therefore, we obtain that
\begin{align*}
J(t, \lambda, h; \rho h)
&= \bar{\eta}\bigl(v-S(h,v)\bigr)\mathbb{E}\left[\int_t^T\lambda_s^{t,\lambda}\rmd s\right] \\
&\quad - \rho h\left(1+\frac{\gamma}{2}\rho h\right)(T-t) 
+ U(h).
\end{align*}
The expectation $\mathbb{E}[\int_t^T\lambda_s^{t,\lambda}\rmd s]$ can be computed by a straightforward adaptation of Proposition \ref{expectation lambda N} (compare also with \cite[Theorem 3.6]{dassios2011dynamic}), thus completing the proof.
\end{proof}

\begin{remark} \label{static strategy}
The lower bound obtained in Proposition \ref{lower bound} is associated to a fixed investment rate which offsets technological obsolescence by maintaining the cybersecurity level constant over time (this follows directly from equation \eqref{sde h}). 
In Section \ref{Comparison with original Gordon-Loeb}, we numerically show that the optimal dynamic investment policy characterized in Theorem \ref{thm:hjb} consistently outperforms any constant investment strategy, highlighting the value of real-time adaptability in cybersecurity investment.
\end{remark}

\section{Numerical Methods} \label{Methods}

In this section, we describe the parameters' choice and the numerical methods adopted for the solution of the optimization problem introduced in Section \ref{The optimization problem}. 

\subsection{Specification of the Model Parameters}
\label{sec:parameters}

We report in Tables \ref{tab: breach function}, \ref{tab: hawkes intensity}, \ref{tab: other params} the {\em standard set} of the model parameters. Unless mentioned otherwise, the numerical a\-na\-ly\-sis will be performed using the standard set of parameters.

\begin{table}[h!]
\centering
   \begin{tabular}{ |c|c|c|c|c|c| } 
 \hline
 function type & $v$ & $a$ & $b$ \\ 
 \hline
  $S_I$ & $0.65$ & $10^{-1}$ & $1$ \\ 
 \hline
 
\end{tabular}
\caption{Specification of the security breach function.}
\label{tab: breach function}
\end{table}
\begin{table}[h!]
    \centering
\begin{tabular}{|c|c|c|c|}
 \hline
    $\alpha$  & $\xi$ & $\beta$ & $\lambda_0$  \\
    \hline
     $27$ & $15$ & $9$ & $27$ \\
      \hline
\end{tabular}
\caption{Parameters of the stochastic intensity.}
\label{tab: hawkes intensity}
\end{table}
\begin{table}[h!]
    \centering
\begin{tabular}{|c|c|c|c|c|c|}
 \hline
    $\delta$ & $\gamma$ & $\bar{\eta}$(k\$) & $U(h)$ & $\rho$ & $T$ \\
    \hline
     $1$ & $0.05$ & $10$ & $\sqrt h$ & $0.2$ & $1$ \\
      \hline
\end{tabular}
\caption{Parameters of the optimization problem.}
\label{tab: other params}
\end{table}

We employ a security breach probability function $S$ of class I, as defined in \eqref{class of S}. 
The parameters values in Table \ref{tab: breach function} are consistent with those determined in \cite{MN2017} through a calibration based on synthetic data and also used as benchmark inputs in subsequent works  (see, e.g., \cite{skeoch2022expanding,mazzoccoli2020robustness}).
Taking $v,a,b$ as in Table \ref{tab: breach function}, the function $h\mapsto S_I(h,v)$ is plotted in Figure \ref{fig: breach function}.

\begin{figure}
    \centering
    \includegraphics[width=1\linewidth]{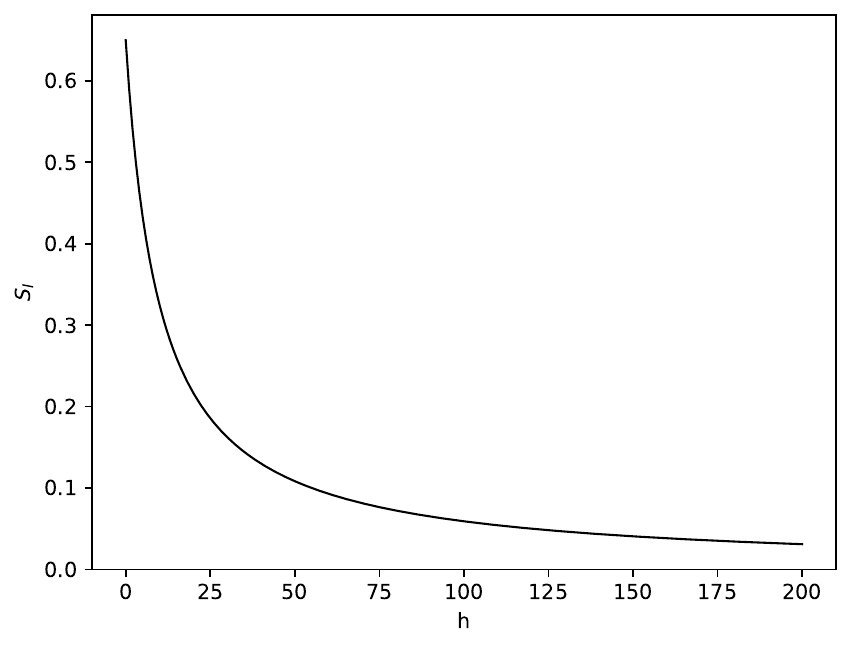}
    \caption{Security breach function (parameters as in Table \ref{tab: breach function}).}
    \label{fig: breach function}
\end{figure}

The parameters of the stochastic intensity of the Hawkes process (see Table \ref{tab: hawkes intensity}) are chosen to generate on average approximately 60 cyberattacks per year. We believe that this is a reasonable figure, in the absence of reliable estimates of the number of cyberattacks targeting a single entity.\footnote{Empirical estimates of the intensity of cyberattacks can be found in the recent works \cite{bessy2021multivariate}, \cite{boumezoued2023cyber}, \cite{li2023modelling}. However, these estimates are not suitable for our purposes, since they are based on the number of attacks at a worldwide scale, while our model takes the viewpoint of a single entity. 
An average of 60 cyberattacks per year is in line with the number of attacks per organisation reported in the Cyber Readiness Report 2024 by Hiscox (see \cite{HISCOX24}).}

\begin{remark}
For the standard set of parameters, we have $\lambda_0=\alpha$ and so the stochastic intensity $\lambda_t$ can be expressed as follows:
\begin{equation}\label{eq:lambda_new}
\lambda_t =  \lambda_0+ \beta\sum_{i=1}^{N_t} e^{-\xi(t-\tau_i)}.
\end{equation}
\end{remark}

We consider a one-year planning horizon ($T=1$) and set an average loss of $10\text{k}\$ $ for each successful breach, resulting in a total expected annual loss of approximately $390 \text{k}\$ $ without cybersecurity investments, which is in the same order of magnitude of \cite{skeoch2022expanding}.
The depreciation rate is set at $\rho=0.2$, consistently with the technological depreciation rates considered in \cite{krutilla2021benefits}. The parameter $\gamma$ is set at a rather low value, in order to avoid an excessive penalization of large investment rates. Finally, we choose $U(h)=\sqrt{h}$, representing a strictly increasing and concave CRRA utility function.

\subsection{Numerical Solution of the HJB-PIDE}
\label{numerics PIDE}

As shown in Section \ref{The optimization problem}, determining the optimal cybersecurity investment requires the solution of the non-linear PIDE \eqref{hjb obsolescence}. Due to the complexity of the problem, one cannot expect to find explicitly an analytical solution and, hence, numerical methods are required. We opt for the {\em method of lines}, as described in \cite{yuan1999ode}. This technique consists in discretizing the PIDE in the spatial domain  $(\lambda,h)\in(0,\infty)\times\mathbb{R}_+$ but not in time, and then in integrating the semi-discrete problem as a system of ODEs. In our setting, we discretize the $(\lambda,h)$ dimensions with a central difference and then numerically solve the resulting ODE system. Similarly to the case of PIDEs arising in L\'evy models (see, e.g., \cite{cont2005finite}), the unbounded space domain $(0,\infty)\times\mathbb{R}_+$ is localized into a bounded domain $[\lambda_{\text{min}},\lambda_{\text{max}}]\times[h_{\text{min}},h_{\text{max}}]$.
We refer to Algorithm \ref{algorithm pide} for a precise description of the implementation of this method. 

In our implementation, we specify as follows the algorithm's meta-parameters:

\begin{table}[h!]
    \centering
\begin{tabular}{|c|c|c|c|c|c|c|c|}
 \hline
     $\lambda_{\text{min}}$ 
 & $\lambda_{\text{max}}$ & $\Delta \lambda$ & $h_{\text{min}}$ & $h_{\text{max}}$ & $\Delta h$  \\
    \hline
    $27$ & $216$ & $1$ & $0$ & $50$ & $0.5$ \\
      \hline
\end{tabular}
\caption{Meta-parameters for Algorithm \ref{algorithm pide}.}
\label{tab: algo params}
\end{table}

The value  $h_{\text{min}} = 0$ corresponds to the absence of cybersecurity investment, while $h_{\text{max}} = 50$ represents an upper bound which is rarely achieved in our setup under the standard parameter set.
We choose $\lambda_{\text{min}}=\lambda_0$, which coincides with the lower bound of the stochastic intensity $\lambda_t$, see equation \eqref{eq:lambda_new}. We set $\lambda_{\text{max}}=\mathbb{E}[\lambda_T] + 7 \sqrt{\text{Var}(\lambda_T)} \approx 216$, in order to ensure that the truncation of the intensity domain does not have any material impact on our numerical results.
The value function $V$ is extrapolated beyond $[\lambda_{\text{min}}, \lambda_{\text{max}}]$ by setting 
\[
V(t,\lambda,h) = V(t,\lambda_{\text{max}},h),
\qquad\text{ for all }
\lambda > \lambda_{\text{max}},
\]
analogously to the scheme implemented in \cite[Section 5.1]{ly2024optimal}. When plotting the function $V$ in Section \ref{Numerical results}, we shall consider a subinterval of $[\lambda_{\text{min}}, \lambda_{\text{max}}]$: intensity values close to $\lambda_{\text{max}}$ are rarely achieved and might lead to numerical instabilities.

\begin{algorithm}
\caption{Numerical solution of the PIDE \eqref{hjb obsolescence}}
    \begin{algorithmic}[1]
        \State Set  $\lambda_{\text{min}}, \lambda_{\text{max}}$, $h_{\text{min}}, h_{\text{max}}$.
        \State Discretize $[\lambda_{\text{min}}, \lambda_{\text{max}}]$, with $\lambda_0=\lambda_{\text{min}}, \lambda_N=\lambda_{\text{max}}$ and $\lambda_n-\lambda_{n-1}=\Delta \lambda$, for $n=1,\ldots,N$.
        \State Discretize $[h_{\text{min}}, h_{\text{max}}]$, with $h_0=h_{\text{min}}, h_M=h_{\text{max}}$ and $h_m-h_{m-1}=\Delta h$, for $m=1,\ldots,M$.
        \State Set $V_{n,m}(t):=V(t, \lambda_n, h_m)$, for all $n$ and $m$.
        \State Approximate the partial derivatives w.r.t. $\lambda$:
        \begin{align*}
\frac{\partial V}{\partial \lambda}(t, \lambda_n, h_m) &\approx \frac{V_{n+1,m}(t)-V_{n-1,m}(t)}{2\Delta \lambda},\\
\frac{\partial V}{\partial \lambda}(t, \lambda_0, h_m) &\approx \frac{V_{1,m}(t)-V_{0,m}(t)}{\Delta \lambda},\\
\frac{\partial V}{\partial \lambda}(t, \lambda_N, h_m) &\approx \frac{V_{N,m}(t)-V_{N-1,m}(t)}{\Delta \lambda}.
        \end{align*}
         \State Approximate the partial derivatives w.r.t. $h$:
         \begin{align*}
\frac{\partial V}{\partial h}(t, \lambda_n, h_m) &\approx \frac{V_{n,m+1}(t)-V_{n,m-1}(t)}{2\Delta h},\\
\frac{\partial V}{\partial h}(t, \lambda_n, h_0) &\approx \frac{V_{n,1}(t)-V_{n,0}(t)}{\Delta h},\\
\frac{\partial V}{\partial h}(t, \lambda_n, h_M) &\approx \frac{V_{n,M}(t)-V_{n,M-1}(t)}{\Delta h}.
         \end{align*}
        \State Let $\tilde{n}=\frac{\lfloor \beta \rfloor}{\Delta \lambda}$ and set \[V(t,\lambda_n+\beta,h_m)\approx V_{(n+\tilde{n}) \wedge N,m}(t).\]
        \State Solve the ODE system given for all $n,m$ by
        \begin{dmath*}
            V'_{n,m}(t)=\xi(\lambda_n-\alpha)\frac{V_{n+1,m}(t)-V_{n-1,m}(t)}{2\Delta \lambda}+\rho h {\frac{V_{n,m+1}(t)-V_{n,m-1}(t)}{2\Delta h}} - \lambda_n(V_{n+\tilde{n}\wedge N,m}(t)-V_{n,m}(t)) - \bar{\eta}(v-S(h_m,v))\lambda_n - \frac{\left(\left(\frac{V_{n,m+1}(t)-V_{n,m-1}(t)}{2\Delta h}-1\right)^+\right)^2}{2\gamma}, 
        \end{dmath*}
   $V_{n,m}(T)= U(h_m).$
    \end{algorithmic}
    \label{algorithm pide}
\end{algorithm}

We have implemented Algorithm \ref{algorithm pide} in Python, using the built-in ODE solver scipy.integrate.solve\_ivp.
We make use of an implicit Runge-Kutta method of the Radau IIA family of order 5 (see \cite{hairer1993solving} for further details). 
The computations were performed on an Intel Xeon E5520 CPU equipped with 32 GB of RAM, requiring approximately 23 hours to compute the value function and the corresponding optimal control.
To empirically assess the convergence of Algorithm \ref{algorithm pide}, we verified that standard discretization error metrics decrease consistently as the spatial grid in $(\lambda, h)$ is refined.

\subsection{Optimal Cybersecurity Investment Rate}

Besides determining the optimal net benefit of cybersecurity investments, we aim at computing the real-time adaptive strategy that best responds to the arrival of cyberattacks. To this end, after solving the PIDE \eqref{hjb obsolescence} via Algorithm \ref{algorithm pide}, we compute numerically the optimal investment rate given in equation \eqref{optimal control} along a simulated sequence of cyberattacks. This entails simulating a trajectory of the stochastic intensity $(\lambda_t(\omega))_{t\in[t_{\text{init}},T]}$, starting from an initial cybersecurity level $H_{\text{init}}$ at time $t_{\text{init}}$. Our numerical method for the computation of the optimal investment rate is described in Algorithm \ref{algorithm control intensity trajectory} and will be numerically implemented in Sections \ref{Gain computation} and \ref{Optimal control along a trajectory}.

\begin{algorithm}
\caption{Numerical computation of the optimal control}
\begin{algorithmic}[1]
\State Set $t_{\text{min}}, t_{\text{max}}$, $\lambda_{\text{min}}, \lambda_{\text{max}}$, $h_{\text{min}}, h_{\text{max}}$.
\item Discretize $[t_{\text{min}}, t_{\text{max}}]$, with $t_0=t_{\text{min}}, t_I=t_{\text{max}}$ and $ t_i-t_{i-1}=\Delta t$, for $i=1,\ldots,I$.
        \State Discretize $[\lambda_{\text{min}}, \lambda_{\text{max}}]$, with $\lambda_0=\lambda_{\text{min}}, \lambda_N=\lambda_{\text{max}}$ and $\lambda_n-\lambda_{n-1}=\Delta \lambda$, for $n=1,\ldots,N$.
        \State Discretize $[h_{\text{min}}, h_{\text{max}}]$, with $h_0=h_{\text{min}}, h_M=h_{\text{max}}$ and $h_m-h_{m-1}=\Delta h$, for $m=1,\ldots,M$.
\State Compute $V(t_i, \lambda_n, h_m)$ and $z^*(t_i, \lambda_n, h_m)$, for $i=0, \ldots, I$,  $n=0, \ldots N$,  $m=0, \ldots, M$.
\State Simulate a trajectory $\lambda_{t_i}(\omega)$, $i=0,\ldots,I$.
\State For the initial time $t_{\text{init}}\geq t_{\text{min}}$, set $\bar i:=\text{argmin}_i \{|t_i-t_{\text{init}}|\}$.
\State Consider the initial state $H_{t_{\bar i}}=H_{\text{init}}$:
\For{$i$ in $\bar{i}, \dots, I$,}
{\State set $k:=\text{argmin}_{\ell} \{|\lambda_{\ell} -\lambda_{t_i}(\omega)|\}$;
\State set 
$j:=\text{argmin}_m \{|h_m-H_{t_i}|\}$;
\State let $z_{t_i}^*=z^*(t_i, \lambda_k, h_j)$;
\State $H^{z^*}_{t_{i+1}}:=H_{t_i}^{z^*}-\rho H_{t_i}^{z^*} \Delta t + z^*_{t_i} \Delta t$.\EndFor}
\end{algorithmic}
\label{algorithm control intensity trajectory}
\end{algorithm}

\section{Results and Discussion} \label{Numerical results}

In this section, we report some numerical results that illustrate  the key properties and implications of the model. In particular, we are interested in assessing the benefit of adopting the optimal dynamic cybersecurity investment policy.

\subsection{Value Function and Optimal Cybersecurity Policy} \label{Value function and optimal control}

Figure \ref{fig: value function optimal control} displays the value function $V$ and the optimal cybersecurity investment rate $z^*$. 
In panels \ref{fig:vf standard lambda fixed} and \ref{fig:control standard lambda fixed} we plot, respectively, $V$ and $z^*$ for fixed intensity $\lambda=27$, varying $t$ and $h$. Coherently with Remark \ref{rem:propertiesV}, we observe that the value function is increasing in $h$, while the optimal investment rate is decreasing. This behavior reflects the fact that higher cybersecurity levels yield greater benefits and reduce the need for further cybersecurity investments. 
In panels \ref{fig:vf standard h fixed} and \ref{fig:control standard h fixed} we plot, respectively, $V$ and $z^*$ for fixed $h=0$, varying $t$ and $\lambda$. Coherently with Remark \ref{rem:propertiesV}, we observe that both the value function and the optimal investment rate are increasing in $\lambda$. This is explained by the fact that, in the presence of a higher risk of cyberattacks, investing in cybersecurity becomes more valuable due to the larger potential of mitigating expected losses.   
As can be seen from panels \ref{fig:vf standard lambda h fixed} and \ref{fig:control standard lambda h fixed}, both the value function and the optimal investment rate decrease over time. This is due to the fact that, under the standard parameter configuration (see Section \ref{sec:parameters}), the residual utility $U(H_T)$ of cybersecurity plays a relatively minor role and, therefore, the value of additional cybersecurity investment declines as the end of the planning horizon $[0,T]$ approaches. 

\begin{figure*}
    \centering
    \subfloat[Value function $V(t, \lambda, h)$ for $\lambda=27$.\label{fig:vf standard lambda fixed}]{\includegraphics[scale=0.5]{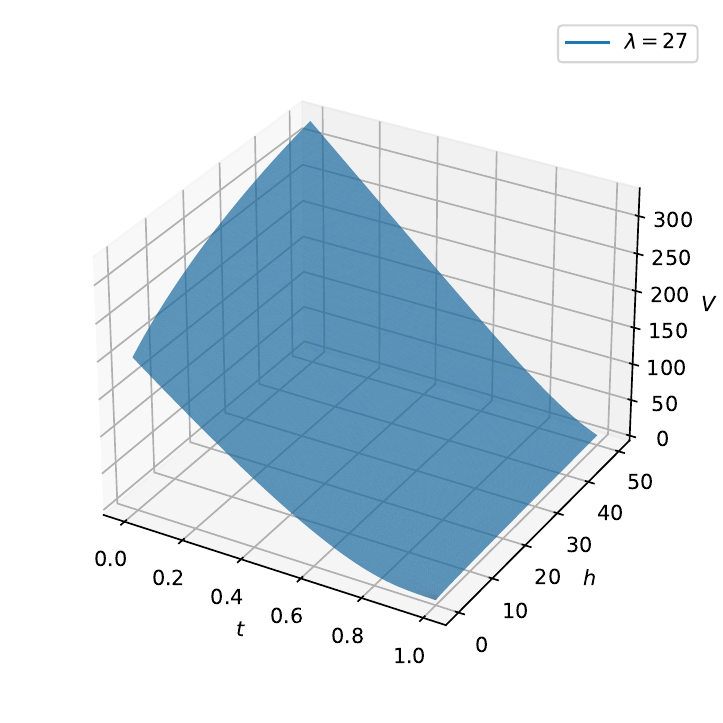}}
    \subfloat[Value function $V(t, \lambda, h)$ for $h=0$.\label{fig:vf standard h fixed}]{\includegraphics[scale=0.5]{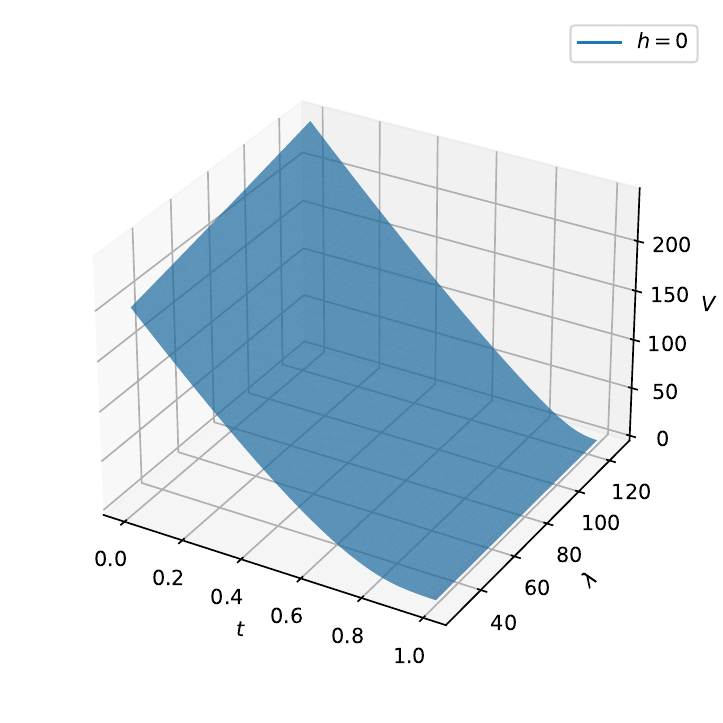}}\subfloat[Value function $V(t, \lambda, h)$ for $\lambda=27, h=0$.\label{fig:vf standard lambda h fixed}]
    {\raisebox{10pt}{\includegraphics[scale=0.35]{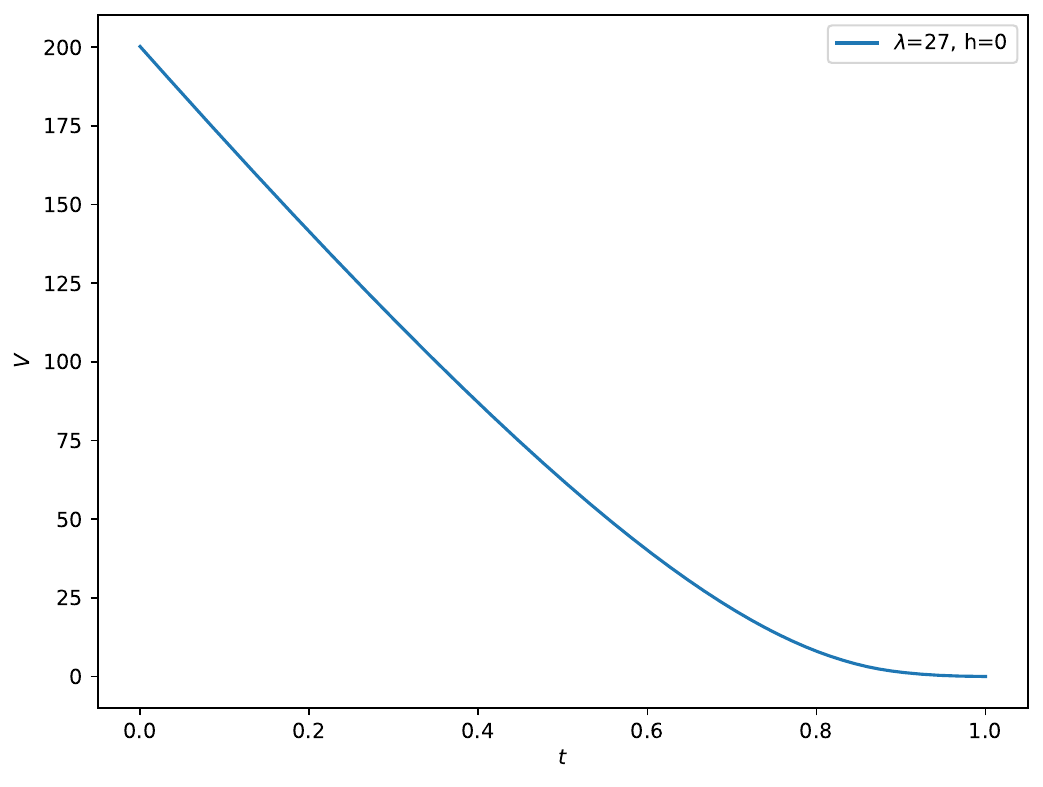}}}\\
    \vspace{-7pt}
    \centering
    \subfloat[Optimal control $z^*_t( \lambda, h)$ for $\lambda=27$.\label{fig:control standard lambda fixed}]{\includegraphics[scale=0.5]{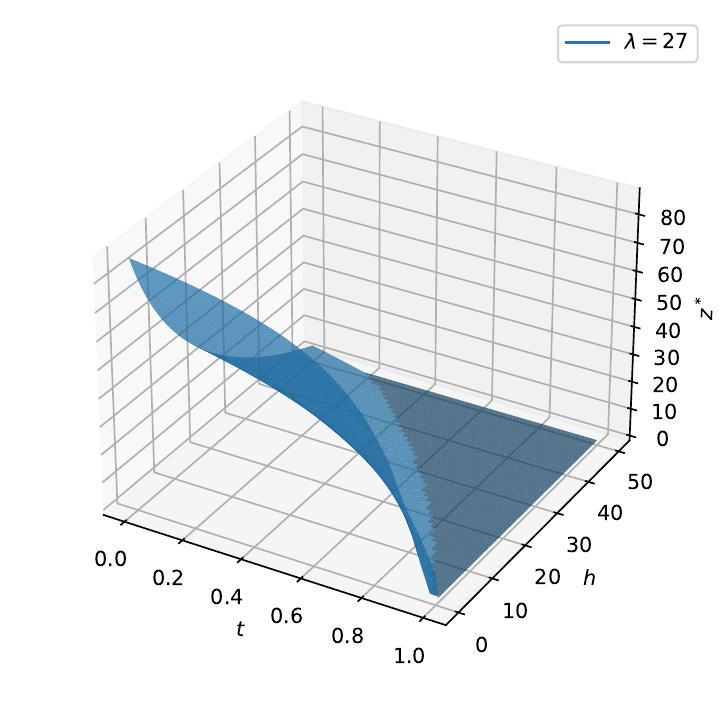}}
    \subfloat[Optimal control $z^*_t( \lambda, h)$ for $h=0$.\label{fig:control standard h fixed}]{\includegraphics[scale=0.5]{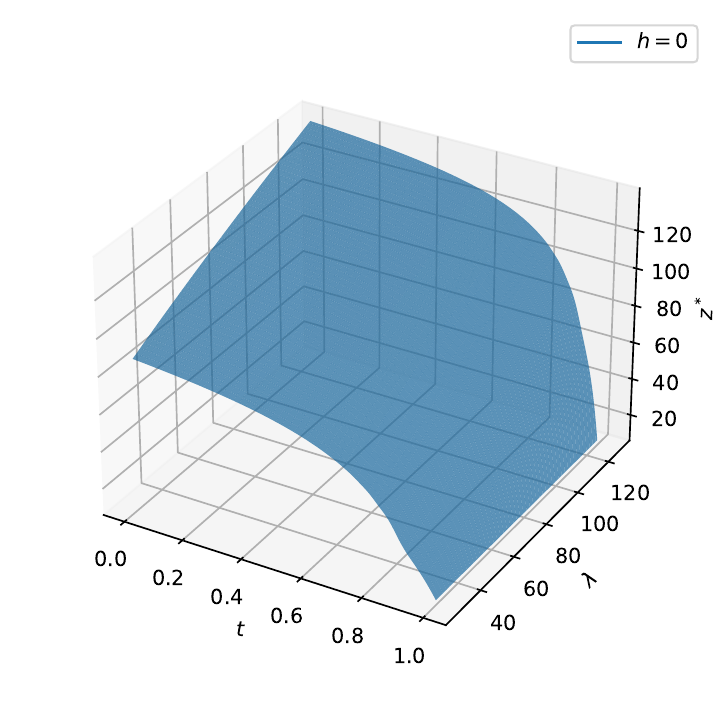}}
    \subfloat[Optimal control $z^*_t( \lambda, h)$ for $\lambda=27, h=0$.\label{fig:control standard lambda h fixed}]{\includegraphics[scale=0.35]{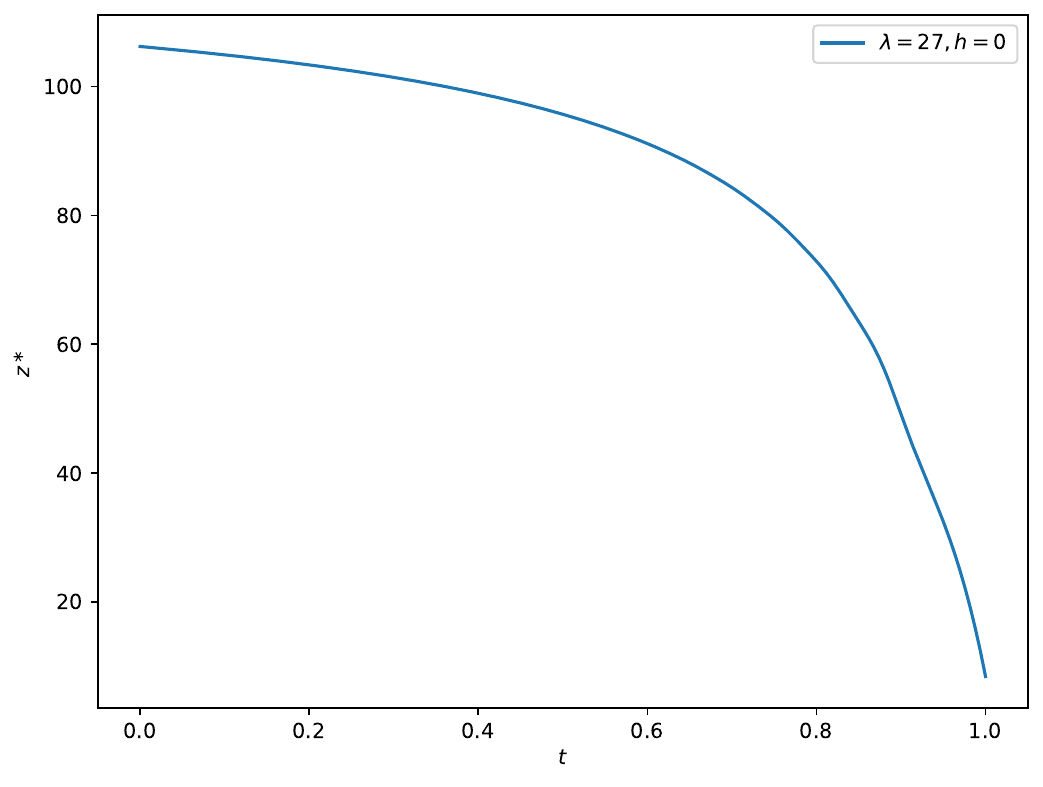}}
    \caption{Value function and optimal investment rate computed under the standard parameters set.}
    \label{fig: value function optimal control}
\end{figure*}

\subsection{Parameter Sensitivity}

In the proposed model, the parameter $\xi$ determines the clustering behavior of cyberattacks. Specifically, higher values of $\xi$ correspond to a more rapid decay of the intensity following each attack, thereby reducing the likelihood of temporally clustered attacks. To evaluate the impact of clustered cyberattacks, we compare in Figure \ref{fig: value function otpimal control varying xi} the value function and the optimal investment rate under two scenarios: $\xi=15$ (more clustered attacks) and $\xi=50$ (less clustered attacks). We observe that both the value function and the optimal investment rate are substantially greater in the case $\xi=15$: if cyberattacks occur in clustered patterns, it is optimal to invest more in cybersecurity in order to mitigate the risk of large cumulative losses arising from rapid attack sequences. This finding underscores the critical importance of accounting for clustering dynamics in the optimal management of cyber-risk.

We also analyze the role of obsolescence in cybersecurity investment decisions, motivated by the analysis in \cite{krutilla2021benefits}, which highlights its significance in a dynamic setup. Figure \ref{fig: value function optimal control comp rho} displays the value function and the optimal investment rate under two contrasting depreciation scenarios: $\rho=0$ (no obsolescence) and $\rho=1$ (high obsolescence).\footnote{In practical terms, a depreciation rate of $\rho = 1$ implies that a given initial cybersecurity level $H_0$ depreciates by over $73\%$ over a one-year period.} We observe that $\rho=1$ leads to smaller values for $V$ and $z^*$, in line with the findings of  \cite{krutilla2021benefits} in a deterministic setup. Our results confirm that a rapid depreciation of cybersecurity effectiveness reduces both the expected net benefit and the incentive to invest.

\begin{figure*}
\centering
    \subfloat[Value function for different $\xi$, $\lambda=27, h=0$.\label{fig: value function xi}]{\includegraphics[scale=0.5]{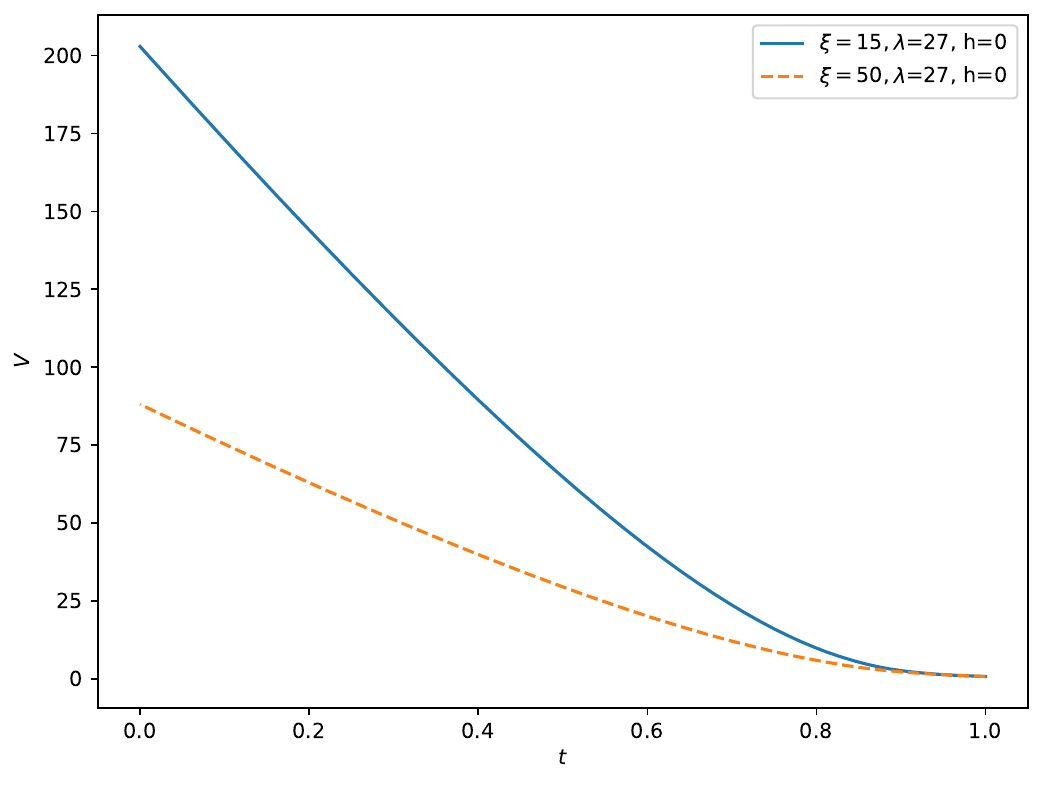}}
    \subfloat[Optimal control for different $\xi$, $\lambda=27, h=0$.\label{fig: control xi}]{\includegraphics[scale=0.5]{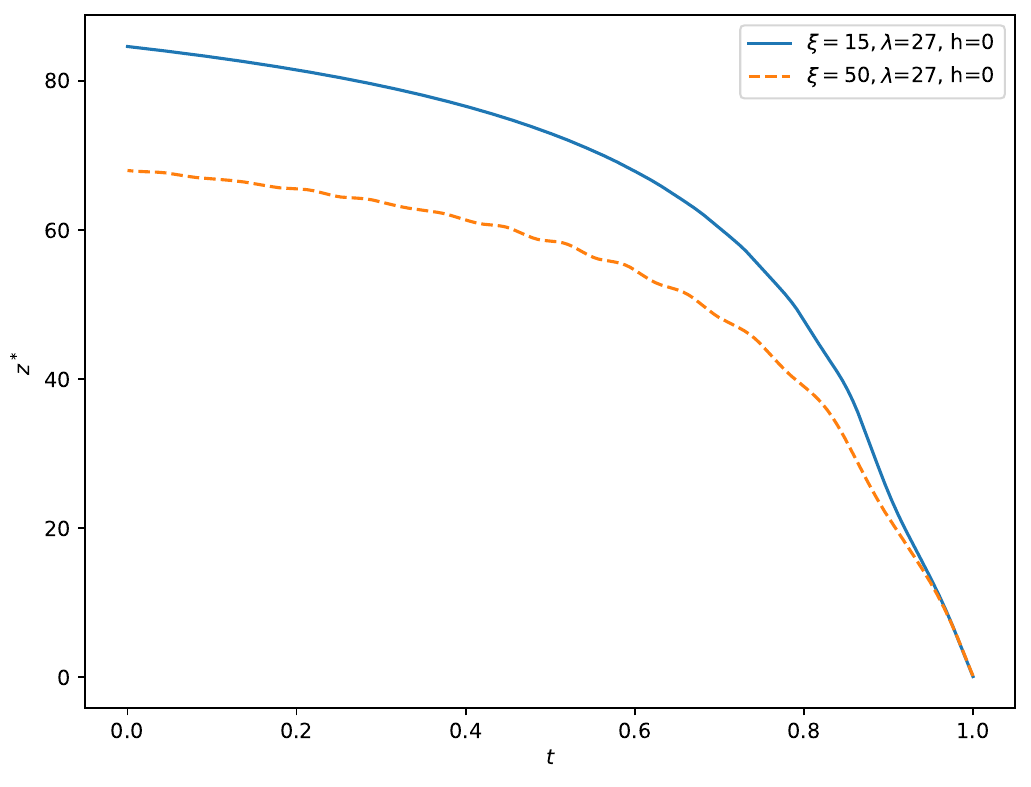}}
    \caption{Value function and optimal investment rate for $\xi=15$ and $\xi=50$, for fixed $h$ and $\lambda$.}
   \label{fig: value function otpimal control varying xi}
\end{figure*}

\begin{figure*}
    \centering
    \subfloat[Value function $V(t, \lambda, h)$ for $\lambda=27, h=0$.\label{fig:vf rho lambda h fixed}]{\includegraphics[scale=0.5]{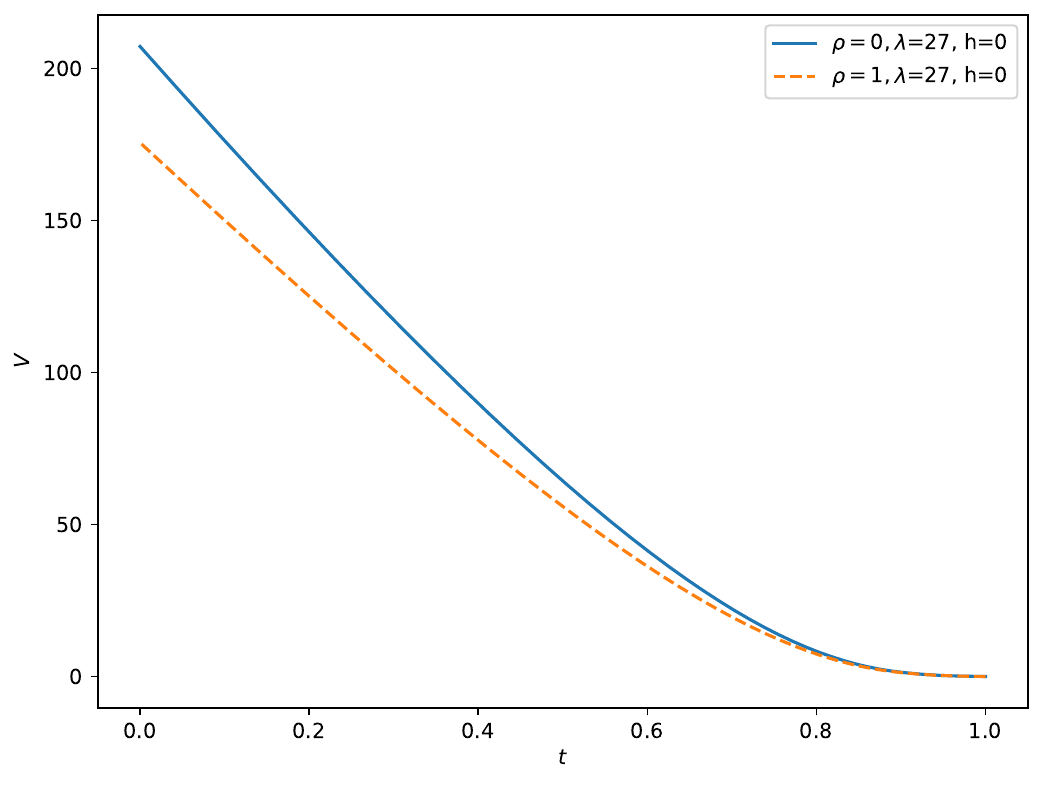}}
    \subfloat[Optimal control $z^*_t( \lambda, h)$ for $\lambda=27, h=0$.\label{fig:control rho lambda h fixed}]{\includegraphics[scale=0.5]{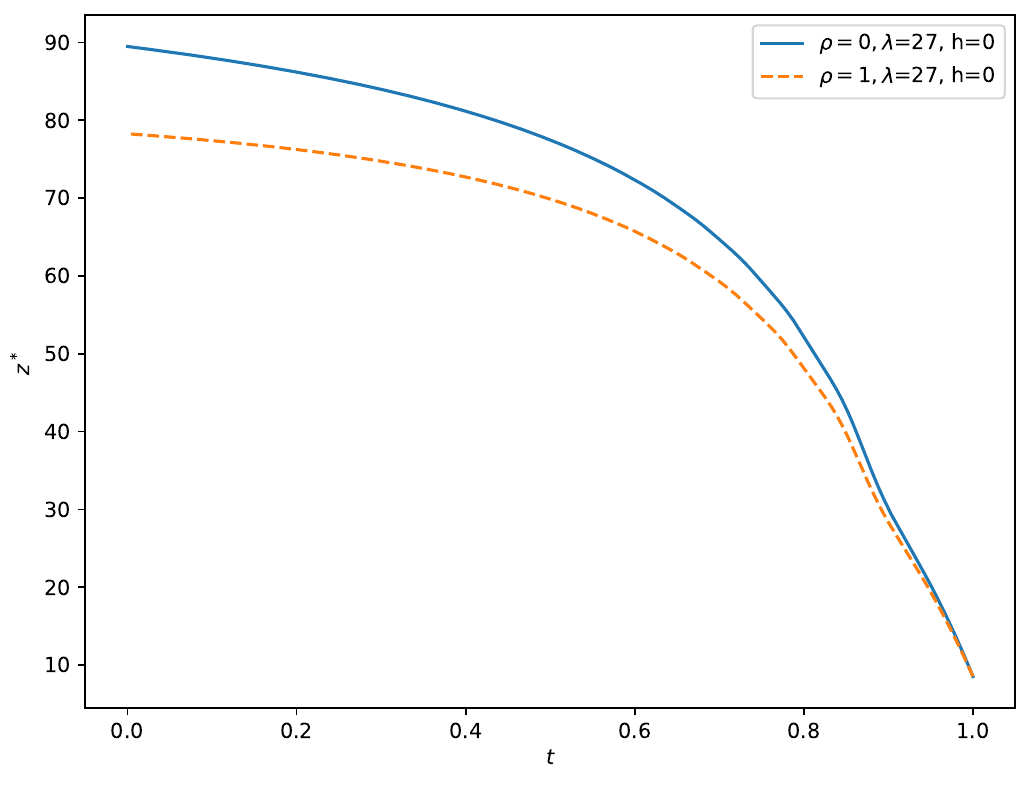}}
    \caption{Value function and optimal investment rate for $\rho=0$ and $\rho=1$, for fixed $h$ and $\lambda$.}
    \label{fig: value function optimal control comp rho}
\end{figure*}

\subsection{Comparison with a Static Investment Strategy} \label{Comparison with original Gordon-Loeb} 

The optimal investment rate $z^*$ characterized in Theorem \ref{thm:hjb} represents the real-time adaptive cybersecurity policy that best responds to the arrival of cyberattacks. 

In order to assess whether the adoption of an adaptive dynamic strategy provides a tangible benefit, we compare it against the best constant investment strategy, i.e., the strategy $z_t=\bar{z}$, for all $t\in[0,T]$, that maximizes the benefit-cost trade-off functional $J$. 
When investing according to a costant rate $\bar{z}$, the benefit-cost trade-off functional $J$ takes the following form:
\begin{equation}\label{J_constant_strategy}\begin{aligned}
J(t,\lambda, h;\bar{z}) &= \int_t^T \bar{\eta}\bigl(v-S(H_s^{t,h;\bar{z}},v)\bigr)\mathbb{E}[\lambda_s^{t,\lambda}]\rmd s\\
&\quad - (T-t)(\bar{z}+\frac{ \gamma}{2}\bar{z}^2) + U(H_T^{t,h;\bar{z}}),
\end{aligned}\end{equation}
where the cybersecurity level $H^{t,h,\bar{z}}$ is given by
\[
H_s^{t,h;\bar{z}}=he^{-\rho(s-t)} + \frac{\bar{z}}{\rho}(1-e^{-\rho(s-t)}).
\]
The optimal constant investment rate $\bar{z}^*$ solves the problem
\begin{equation}\label{constant optimization}
J(t,\lambda, h; \bar{z}^*):= \sup_{\bar{z}\in\mathbb{R}_+}J(t,\lambda, h;\bar{z}).
\end{equation}
In view of Proposition \ref{expectation lambda N}, the expectation $\mathbb{E}[\lambda^{t,\lambda}_s]$ in \eqref{J_constant_strategy} can be computed in closed form. Therefore, the optimization problem \eqref{constant optimization} reduces to a deterministic maximization with respect to a scalar variable, which can be easily solved numerically. To this effect, we adopt the built-in global scalar optimizer scipy.optimize.differential\_evolution in Python.

We quantify the relative gain obtained by investing according to the optimal dynamic policy $z^*$ versus the constant policy $\bar{z}^*$ by computing the following quantity:
\begin{equation}\label{gl gain percentage} 
\% \text{gain}(t,\lambda,h) := 100 \times \frac{V(t, \lambda, h)-J(t, \lambda, h; \bar{z}^*)}{J(t, \lambda, h; \bar{z}^*)}.
\end{equation}
Figure \ref{fig: GL gain constant} displays the relative gain over time for varying cybersecurity levels $h$, for $\lambda=27$ fixed. 
At the initial time $t=0$, the gain reaches $15\%$ for $h=0.5$, $14\%$ for $h=1$, $12\%$ for $h=2$, while it is $9.04\%$ for $h=5$, $5.7\%$ for $h=10$ and $2.6\%$ for $h=20$. 
These results show that the optimal dynamic investment strategy $z^*$ consistently outperforms the best constant strategy $\bar{z}^*$, underscoring the importance of adaptive and responsive cybersecurity investments.
The fact that the gain is rather small for large initial cybersecurity levels is coherent with the findings in Section \ref{Value function and optimal control}: when the initial cybersecurity level is already high, the  benefit of further investments diminishes, thereby reducing the relative advantage of the optimal policy.
Moreover, a further analysis shows that the gain increases monotonically with respect to $\lambda$, indicating that the advantage of adopting the dynamic optimal policy \eqref{optimal control} becomes more pronounced in high-risk scenarios.

\begin{figure}
    \subfloat{\includegraphics[scale=0.45]{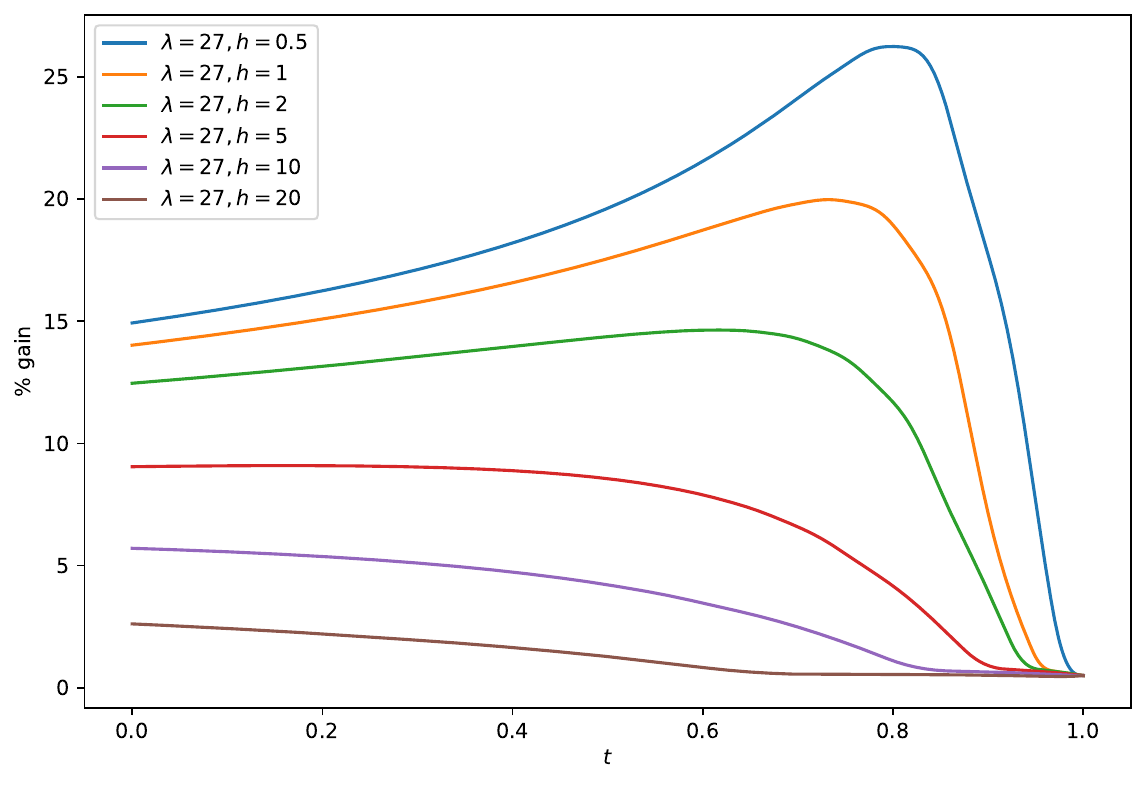}}
    \caption{Relative gain with respect to the optimal constant investment rate.}
   \label{fig: GL gain constant}
\end{figure}

\subsection{Comparison with a Standard Poisson Model} \label{Comparison with Poisson}

To further assess the impact of clustered cyberattacks, we compare our model, which features a self-exciting Hawkes process, with a simplified version based on a standard Poisson process. A Poisson process $P=(P_t)_{t\in[0,T]}$ is characterized by a constant intensity $\lambda^P$ and does not capture any temporal dependence in the arrival of attacks. Indeed, conditionally on $P_T=n$, the attack times are distributed as the order statistics of $n$ i.i.d. random variables uniformly distributed on $[0,T]$, for every $n\in\mathbb{N}$. This setup can be recovered as a special case of the model introduced in Section \ref{A dynamic extension} by setting $\beta=0$ and $\alpha=\lambda_0$ in the intensity dynamics \eqref{intensity-SDE}.

We consider the same optimization problem as in Section \ref{The optimization problem} and we replace the Hawkes process $N$ with a Poisson process $P$ of constant intensity $\lambda^P$ and denote the resulting optimal investment rate by $z^{P*}$. A key observation is that, in this case, problem \eqref{linear criterion new z} reduces to a deterministic optimal control problem. The associated value function, $V^P(t,h)$, solves the following PDE:
\begin{align} 
    &\frac{\partial V^P}{\partial t} -\rho h \frac{\partial V^P}{\partial h} + \lambda^P\bar{\eta}\,(v-S(h,v)) + \frac{\left( (\frac{\partial V^P}{\partial h}-1)^+\right)^2}{2 \gamma}=0, \nonumber\\
    &V^P(T, h) = U(h). 
\label{eq poisson problem}
\end{align}
This PDE can be numerically solved using a scheme similar to Algorithm \ref{algorithm pide}.
Analogously to Theorem \ref{thm:hjb}, the optimal investment rate $z^{P*}$ in the Poisson model is given by
\begin{equation}\label{Poi_opt_control}
z^{P*}=\frac{\left(\frac{\partial V^P}{\partial h}-1\right)^+}{\gamma}.
\end{equation}

\begin{remark}
The optimal policy $z^{P*}$ is deterministic.
This is due to the fact that, in the Poisson model, the occurrence of a cyberattack does not carry any informational content.
\end{remark}

We compare the Hawkes-based model with two Poisson-based benchmarks:
\begin{enumerate}
\item[(i)] a Poisson model with intensity $\lambda^P_b$ chosen as
\begin{equation} \label{lambda poi baseline}
        \lambda^P_b=\lambda_0=27;
\end{equation}
\item[(ii)] a Poisson model with intensity $\lambda^P_e$ chosen as
\begin{equation}\label{lambda poi expectation}
\lambda^P_e=\frac{\lambda_0 \xi }{\xi-\beta} + \frac{1-e^{-\xi T}}{T(\xi-\beta)}\left(\lambda_0-\frac{\lambda_0\xi}{\xi-\beta}\right) \approx 61. 
\end{equation}
\end{enumerate}
The first case corresponds to a Poisson process with the same baseline intensity of the Hawkes process. This scenario can be thought of as the situation where the entity underestimates the likelihood of cyberattacks (possibly due to relying on a limited or unrepresentative dataset) and considers it to be constant over time.
In the second case, in view of Proposition \ref{prop:expectations}, the value $\lambda^P_e$ is chosen so that $\mathbb{E}[P_T]=\mathbb{E}[N_T]$, ensuring that the Hawkes-based model and the Poisson model with intensity $\lambda^P_e$ generate the same expected number of cyberattacks over the planning horizon $[0,T]$. This reflects a case where the average attack frequency is estimated correctly, but the clustering dynamics are ignored.

We shall make use of the following notation:
\begin{itemize}
    \item $V^P_b(t,h)$ is the value function associated to the PDE \eqref{eq poisson problem} for the Poisson model with intensity $\lambda^P_b$ specified in \eqref{lambda poi baseline} and $z_t^{P*, b}(h)$ is the associated optimal control;
  \item $V^P_e(t,h)$ is the value function associated to the PDE \eqref{eq poisson problem} for the Poisson model with intensity $\lambda^P_e$ specified in \eqref{lambda poi expectation} and $z_t^{P*, e}(h)$ is the associated optimal control;
    \item $V(t,\lambda, h)$ is the value function associated to the PIDE \eqref{hjb obsolescence} and $z_t^*(\lambda, h)$ is the associated optimal control.
\end{itemize}

\subsubsection{Value Functions and Optimal Cybersecurity Policies}
\label{value_fcts_Poisson}
Figure \ref{fig: hawkes poisson baseline} displays the results of the comparison with the Poisson model (i) with intensity $\lambda^P_b$. We observe that both the value function and the optimal cybersecurity investment rate under the Hawkes-based model consistently dominate their counterparts in the Poisson model (i) across the entire planning horizon. This is a direct consequence of the fact that $\lambda_t\geq\lambda^P_b$, for all $t\in[0,T]$. In other words, the Poisson model (i) not only disregards the temporal clustering of cyberattacks, but also systematically underestimates their frequency. As a result, the perceived benefit of cybersecurity investment is lower, leading in turn to a suboptimal investment strategy.

\begin{figure*}
    \centering
    \subfloat[Value functions $V(t, \lambda^P_b, h)$ and $V^P_b(t,h)$ .\label{fig:vf baseline lambda fixed}]{\includegraphics[scale=0.5]{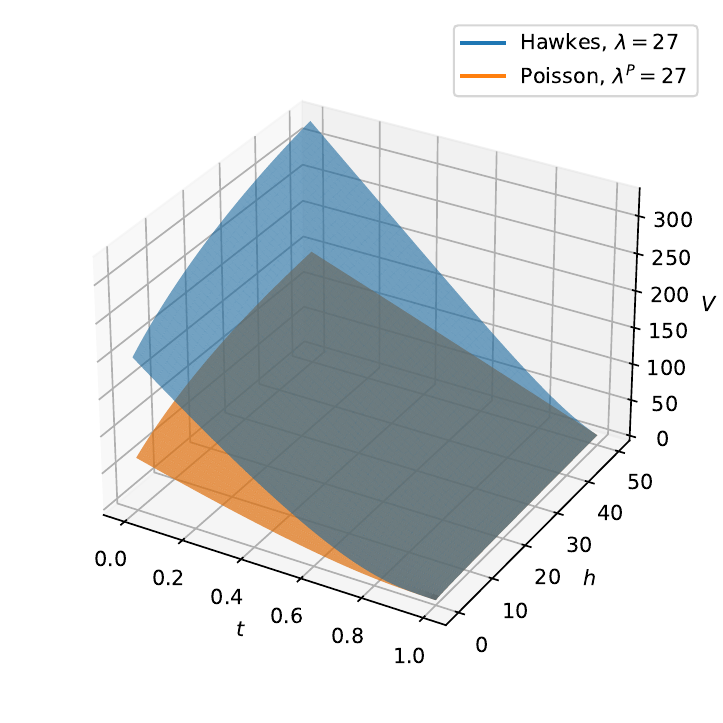}}
    \subfloat[Value functions $V(0, \lambda, h)$ and $V^P_b(0,h)$.\label{fig:vf baseline t fixed}]{\includegraphics[scale=0.5]{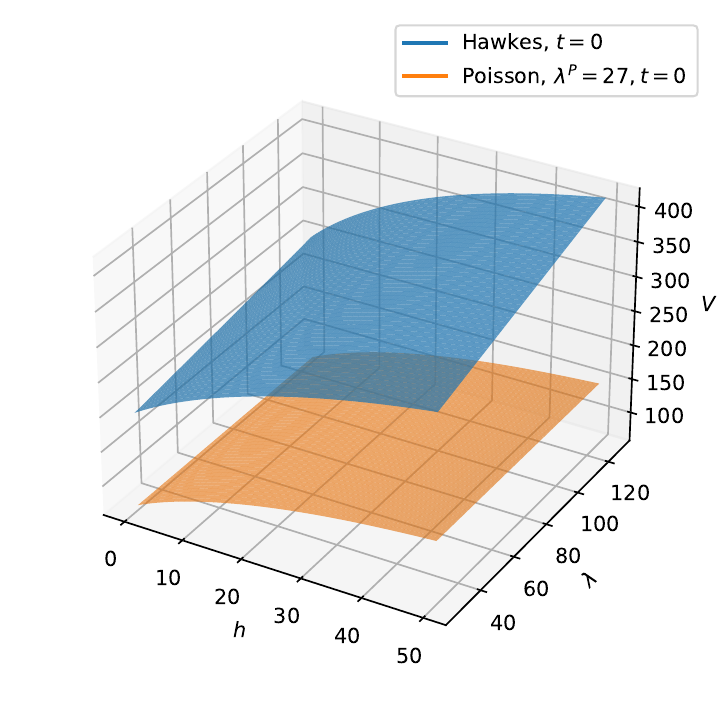}}
    \subfloat[Value functions $V(t, \lambda^P_b, 0)$ and $V^P_b(t,0)$.\label{fig:vf baseline lambda h fixed}]{\includegraphics[scale=0.35]{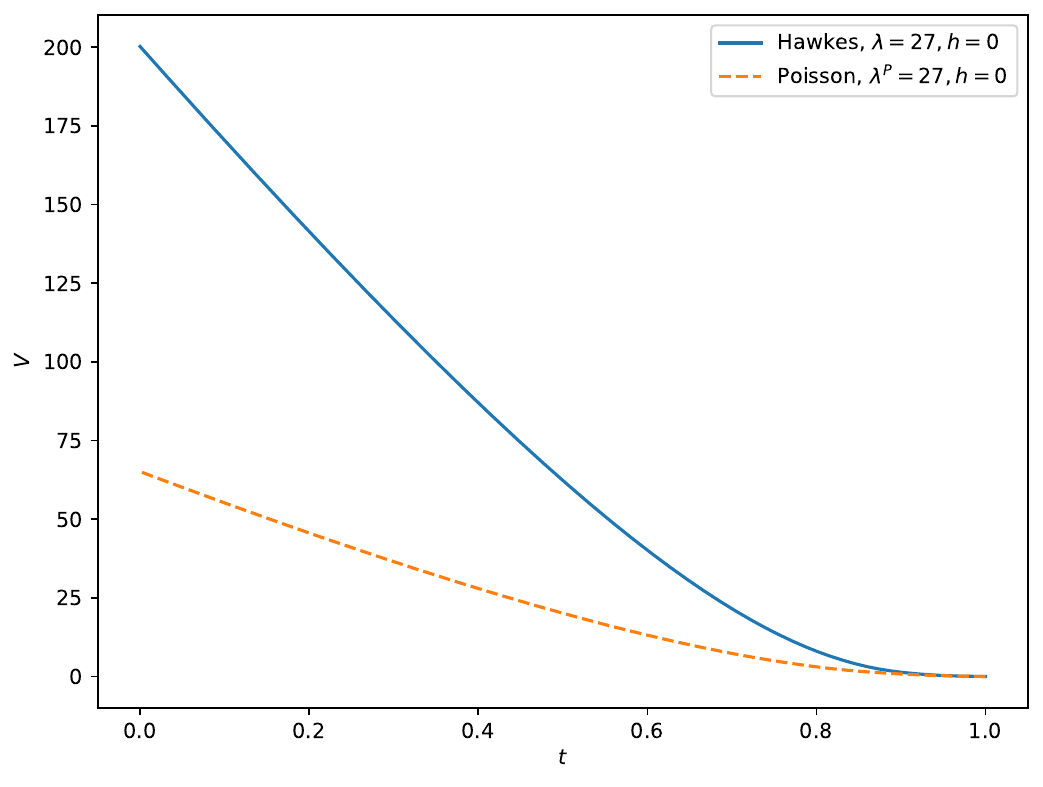}}\\
    \vspace{-7pt}
    \centering
    \subfloat[Optimal controls $z_t^*( \lambda_b^P, h)$ and $z_t^{P^*,b}(h)$.\label{fig:control baseline lambda fixed}]{\includegraphics[scale=0.5]{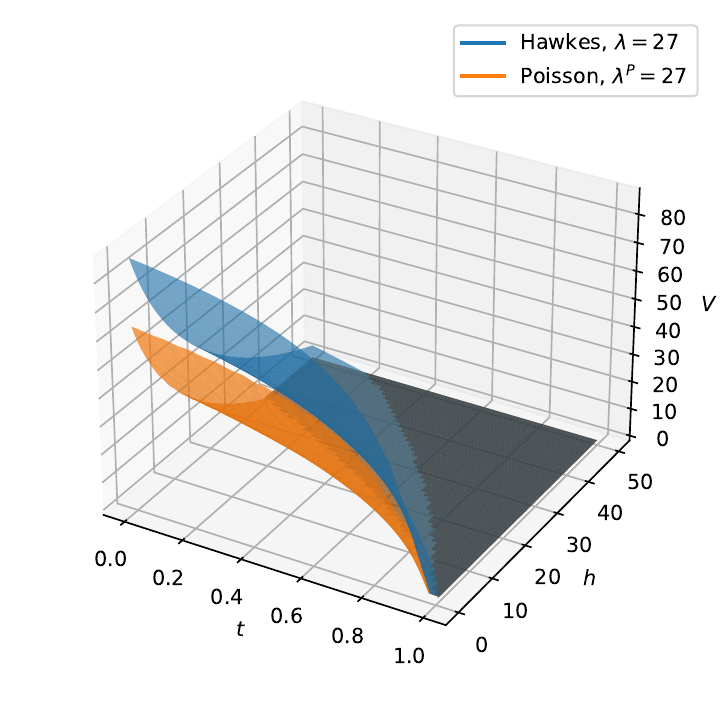}}
    \subfloat[Optimal controls $z_0^*( \lambda_b^P, h)$ and $z_0^{P^*,b}(h)$.\label{fig:control baseline t fixed}]{\includegraphics[scale=0.5]{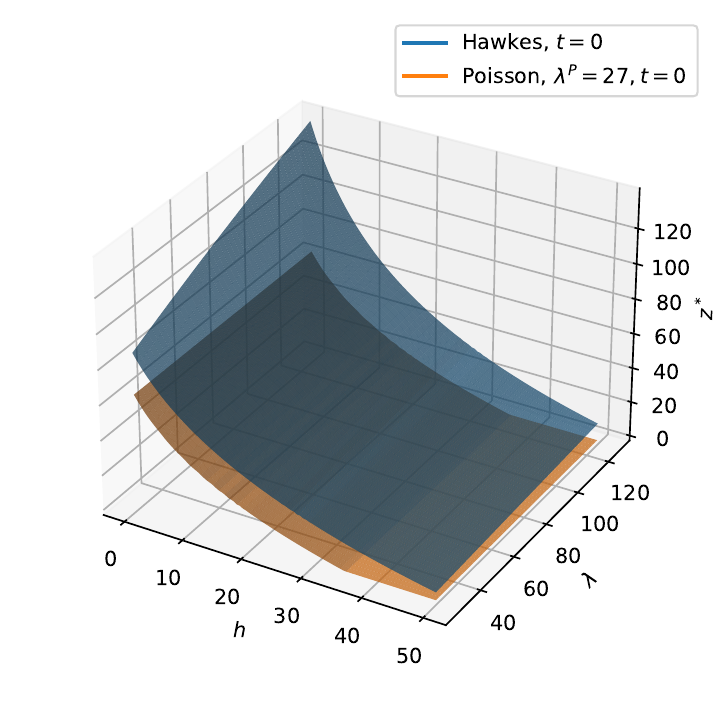}}
    \subfloat[Optimal controls $z_t^*( \lambda_b^P, 0)$ and $z_t^{P^*,b}(0)$.\label{fig:control baseline lambda h fixed}]{\includegraphics[scale=0.35]{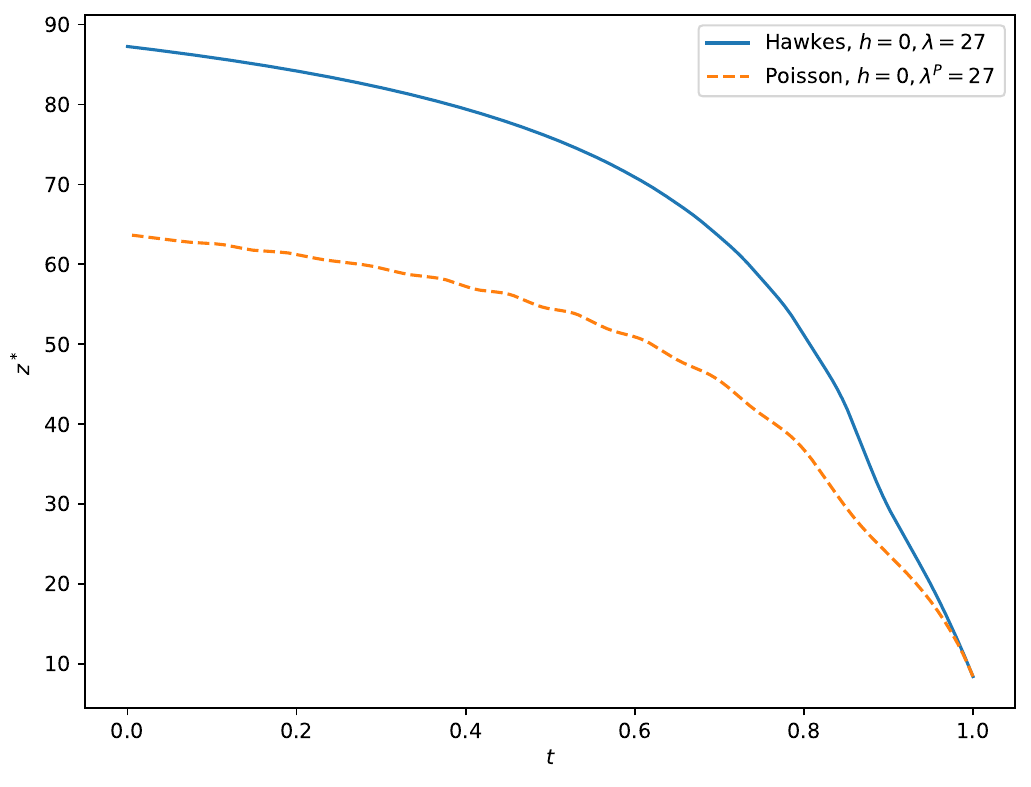}}
    \caption{Comparison with a Poisson model with constant intensity $\lambda^P_b=27$.}
    \label{fig: hawkes poisson baseline}
\end{figure*}

\begin{figure*}
    \centering
    \subfloat[Value functions $V(t, \lambda^P_e, h)$ and $V^P_e(t,h)$. \label{fig:vf exp lambda fixed}]{\includegraphics[scale=0.5]{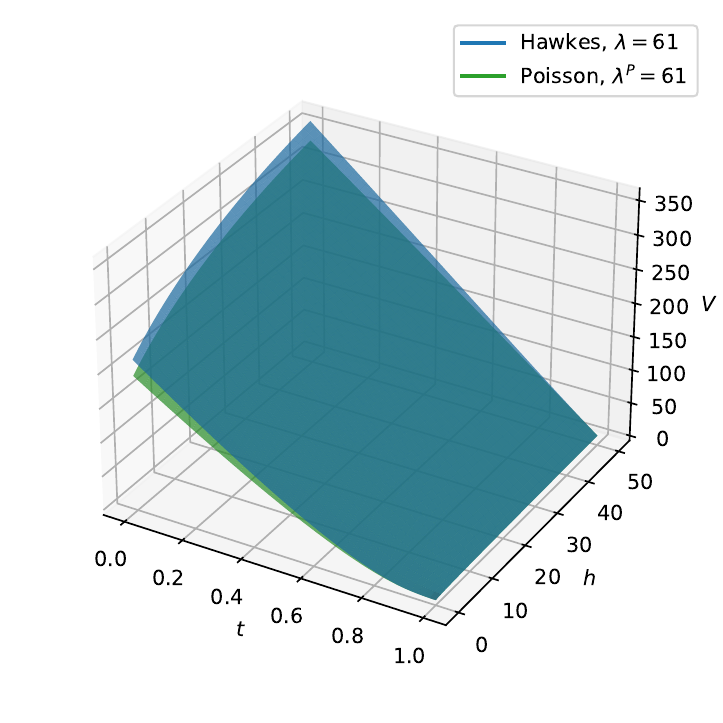}}
    \subfloat[Value functions $V(0, \lambda, h)$ and $V^P_e(0,h)$. \label{fig:vf exp t fixed}]{\includegraphics[scale=0.5]{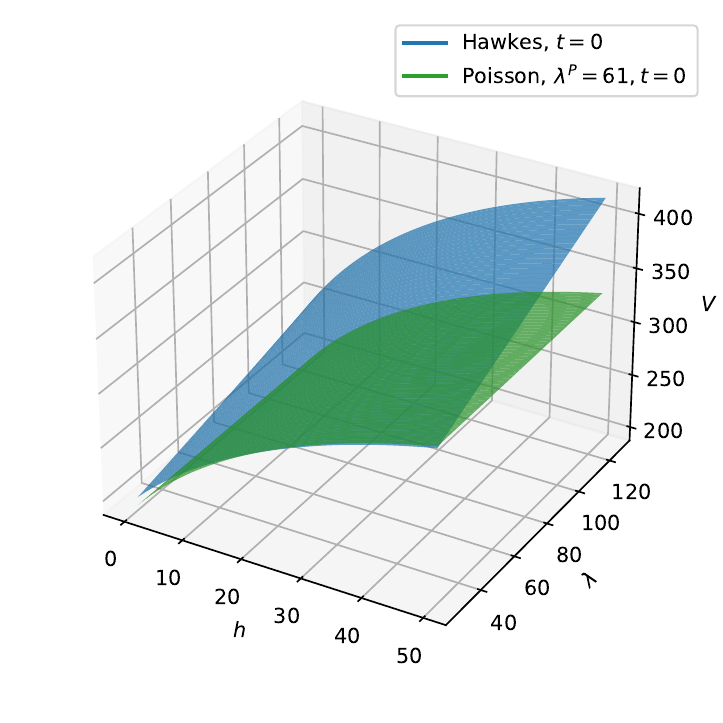}}
    \subfloat[Value functions $V(t, \lambda^P_e, 0)$ and $V^P_e(t,0)$. \label{fig:vf exp lambda h fixed}]{\includegraphics[scale=0.35]{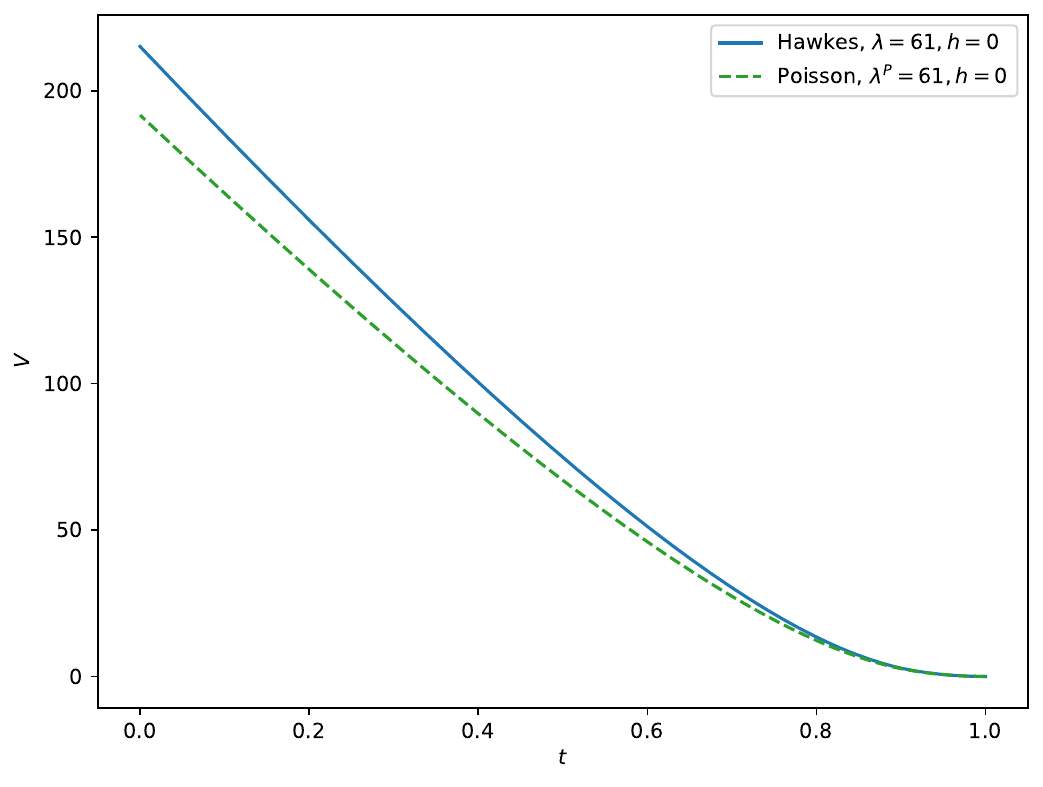}}\\
    \vspace{-7pt}
    \centering
    \subfloat[Optimal controls $z_t^*( \lambda_e^P, h)$ and $z_t^{P^*,e}(h)$. \label{fig:control exp lambda fixed}]{\includegraphics[scale=0.5]{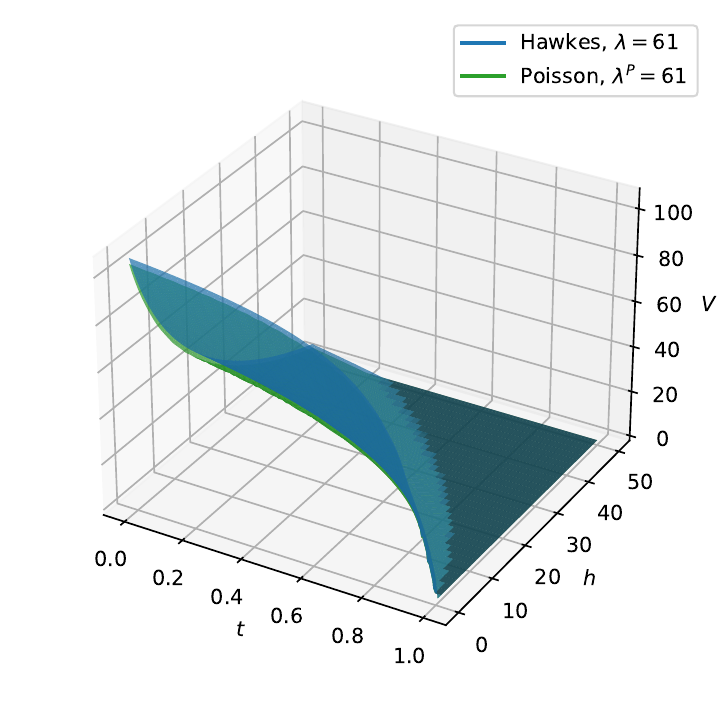}}
    \subfloat[Optimal controls $z_0^*( \lambda, h)$ and $z_0^{P^*,e}(h)$. \label{fig:control exp t fixed}]{\includegraphics[scale=0.5]{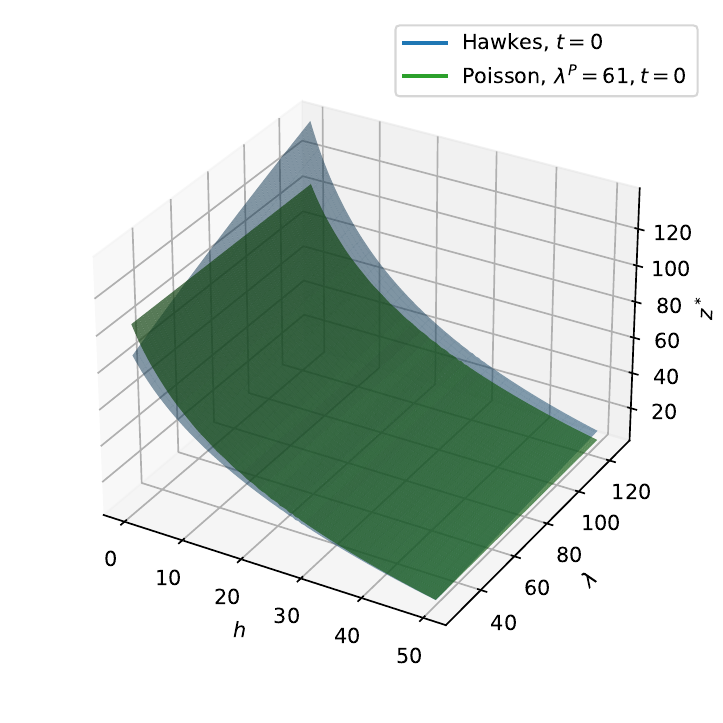}}
    \subfloat[Optimal controls $z_t^*( \lambda_e^P, 0)$ and and $z_t^{P^*,e}(0)$. \label{fig:control exp lambda h fixed}]{\includegraphics[scale=0.35]{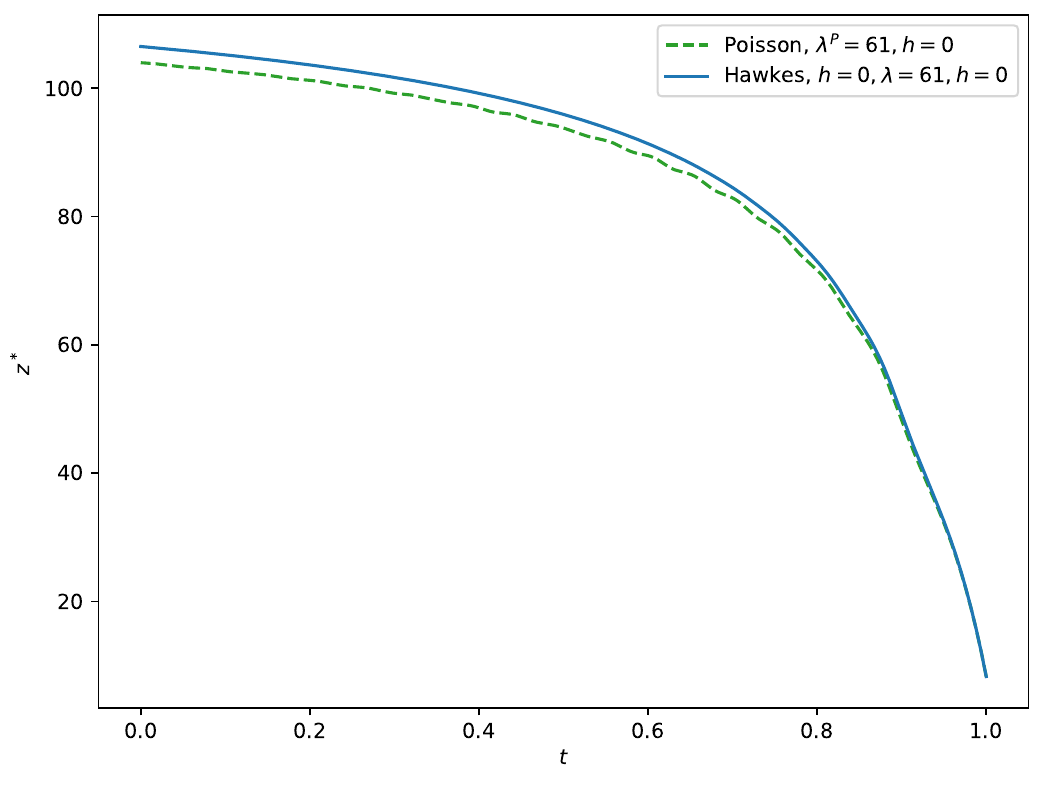}}
    \caption{Comparison with a Poisson model with constant intensity $\lambda^P_e= 61$.}
    \label{fig: hawkes poisson exp}
\end{figure*}

Figure \ref{fig: hawkes poisson exp} reports the comparison with the Poisson model (ii) with intensity $\lambda^P_e$. Panels \ref{fig:vf exp lambda fixed} and \ref{fig:control exp lambda fixed} show that the benefit of cybersecurity investment and the optimal investment rate are slightly greater in the presence of clustered attacks (Hawkes-based model). This finding is confirmed in Panels \ref{fig:vf exp lambda h fixed} and \ref{fig:control exp lambda h fixed}, which compare the value functions and optimal investment rates for fixed values $\lambda=\lambda^P_e$ and $h=0$.
Further insight is provided by panels \ref{fig:vf exp t fixed} and \ref{fig:control exp t fixed}, which display respectively the value functions and the optimal investment rates at the initial time $t=0$, across varying intensity levels. We can observe that the difference between the Hawkes and the Poisson models is negligible for small values of $\lambda$, while it becomes increasingly pronounced at higher values of $\lambda$. Interestingly, panel \ref{fig:control exp t fixed} shows that the optimal investment rate under the Hawkes model may be either higher or lower than that in the Poisson model, depending on whether the current intensity $\lambda$ exceeds $\lambda^P_e$ or not. This feature will be analyzed in more detail in Section \ref{Optimal control along a trajectory} below.
Overall, these findings indicate that even when the average attack intensity is correctly estimated, neglecting the temporal clustering of cyberattacks can lead to suboptimal cybersecurity investment decisions.

\subsubsection{Relative Gain} \label{Gain computation} 

Proceeding similarly to Section \ref{Comparison with original Gordon-Loeb}, we now evaluate the additional benefit derived from implementing the optimal adaptive policy $z^*$, as defined in \eqref{optimal control}, relative to the dynamic but deterministic policy $z^{P*}$ derived in a Poisson-based model. We assume that the underlying model is the one introduced in Section \ref{A dynamic extension} and compute the value function $V$ given in \eqref{value function} via Algorithm \ref{algorithm pide}, using the standard parameter set described in Section \ref{sec:parameters}.
When employing the deterministic strategy $z^{P*}$, the expected net benefit from cybersecurity investment is quantified as follows:
\begin{align*}
J(t, \lambda, h; z^{P*}) & = \int_t^T \left(\bar{\eta} (v-S(H_s^{t,h; z^{P*}},v))\mathbb{E}[\lambda_s^{t,\lambda}]\right)\rmd s \\
&\; - \int_t^T\left(z^{P*}_s+\frac{ \gamma}{2}(z^{P*}_s)^2 \right) \rmd s+ U\bigl(H_T^{t,h;z^{P*}}\bigr),
\end{align*}
where $H^{t,h; z^{P*}}_s$ is defined as in \eqref{h markovian} with $z=z^{P*}$ and $\mathbb{E}[\lambda_s^{t,\lambda}]$ can be computed explicitly by Proposition \ref{expectation lambda N}. Once the PDE \eqref{eq poisson problem} is numerically solved, $z^{P*}$ can be computed via Algorithm \ref{algorithm control intensity trajectory} taking $\lambda$ constant. Our numerical implementation of the Poisson-based model adopts the following specification:

\begin{table}[h!]
    \centering
\begin{tabular}{|c|c|c|c|c|c|c|c|c|c}
 \hline
     $h_{\text{min}}$ & $h_{\text{max}}$ & $\Delta h$ & $\lambda_{\text{min}}$ & $\lambda_{\text{max}}$ & $\Delta \lambda$ & $t_{\text{init}}$ & $H_{\text{init}}$ & $\lambda_t(\omega)$\\
    \hline
    $0$ & $50$ & $0.5$ & $27$ & $216$ & $1$ & $t$ & $h$ & $\lambda^P$ \\
    \hline
\end{tabular}
\caption{Meta-parameters for Algorithm \ref{algorithm control intensity trajectory}.}
\label{tab: algo 2 params gain}
\end{table}

The gain of the optimal investment policy $z^*$ with respect to $z^{P*}$ is computed as follows, in analogy to \eqref{gl gain percentage}:
\begin{equation} \label{poi gain percentage}
\%\text{gain}^P(t,\lambda,h) := 100\times\frac{V(t,\lambda,h) -J(t, \lambda, h; z^{P*})}{J(t, \lambda, h; z^{P*})}.
\end{equation}
Figure \ref{fig: Poi gain} reports the quantity $\%\text{gain}^P(t,\lambda,h)$ comparing the Hawkes-based model against two Poisson-based benchmarks with intensities $\lambda^P_b$ and $\lambda^P_e$, as considered above. 
For $h=0$, the gain increases with $\lambda$, ranging between  $7.6\%$ and $11.4\%$ for $\lambda^P_b$, and between $0.04\%$ and $0.6\%$ for $\lambda^P_e$. For $h=20$, the gain becomes nearly constant in $\lambda$. The fact that the gain for $\lambda^P_e$ is limited can be explained by the fact that the objective functional \eqref{linear criterion u z} is linear with respect to the losses and $\lambda^P_e$ is chosen in such a way that the Poisson-based model generates the same expected losses of our Hawkes-based model.

\begin{figure} 
   \includegraphics[scale=0.6]{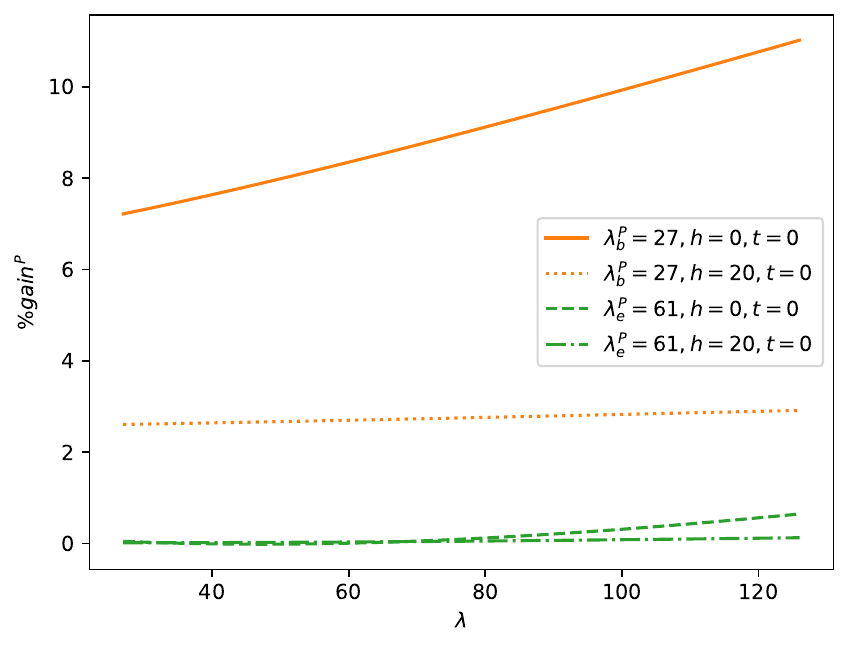}
    \caption{Relative gain with respect to the Poisson deterministic strategy, as defined in \eqref{poi gain percentage}.}
   \label{fig: Poi gain}
\end{figure}

\subsubsection{Adaptive Dynamics of the Optimal Investment Policy} \label{Optimal control along a trajectory}

Finally, we illustrate the adaptive behavior of the optimal investment policy given in equation \eqref{optimal control}. While the overall improvement over a Poisson-based strategy may appear limited in terms of overall gain (see Figure \ref{fig: Poi gain}), the key strength of our approach lies in its capacity to dynamically adjust the cybersecurity investment in response to the arrival of cyberattacks.
To this effect, panels \ref{fig:int traj 1} and \ref{fig:int traj 2} display two simulated paths of the Hawkes intensity $(\lambda_t)_{t\in[0,T]}$, alongside the constant intensities $\lambda^P_b$ and $\lambda^P_e$ defined in Section \ref{Comparison with Poisson}. 
The corresponding optimal investment policies are shown in panels \ref{fig:control traj 1} and \ref{fig:control traj 2}. Consistent with the analysis in Section \ref{value_fcts_Poisson}, the optimal investment rate $z^*_t$ is always larger than the Poisson-based benchmark $z^{P*}_b$, due to the fact that $\lambda_t\geq\lambda^P_b$, for all $t\in[0,T]$.
In contrast, the comparison with the benchmark strategy $z^{P*}_e$ is more nuanced. We highlight in cyan the time intervals during which $\lambda_t > \lambda^P_e$. Our simulations reveal that when $\lambda_t \leq \lambda^P_e$, the adaptive strategy $z^*_t$ closely aligns with $z^{P*}_e$. However, when $\lambda_t > \lambda^P_e$, especially during extended periods resulting from clusters of cyberattacks, the investment rate $z^*_t$ increases markedly, exceeding the corresponding deterministic strategy.
This shows that the optimal cybersecurity investment policy $z^*_t$ can react in real-time to rapid sequences of cyberattacks.
Finally, under the standard parameter set, the investment rate naturally declines toward the end of the planning horizon $[0,T]$, as the accumulated cybersecurity level suffices to mitigate future risk.

\begin{figure*}
    \centering
    \subfloat[Intensity path.\label{fig:int traj 1}]{\includegraphics[scale=0.5]{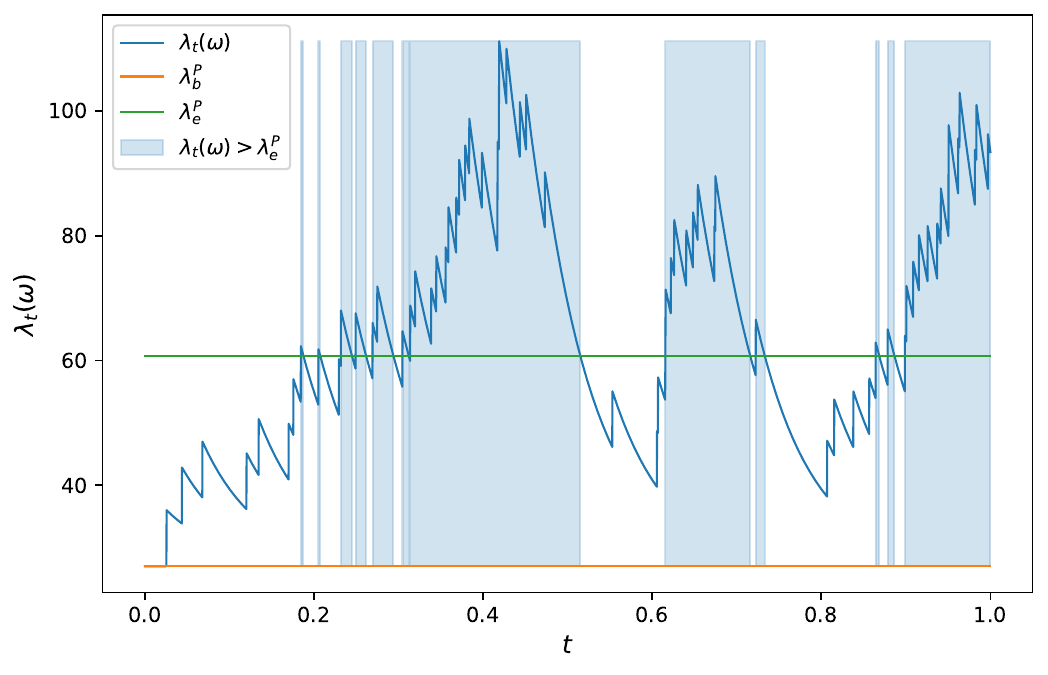}}
    \subfloat[Optimal strategies along the simulated intensity path.\label{fig:control traj 1}]{\includegraphics[scale=0.5]{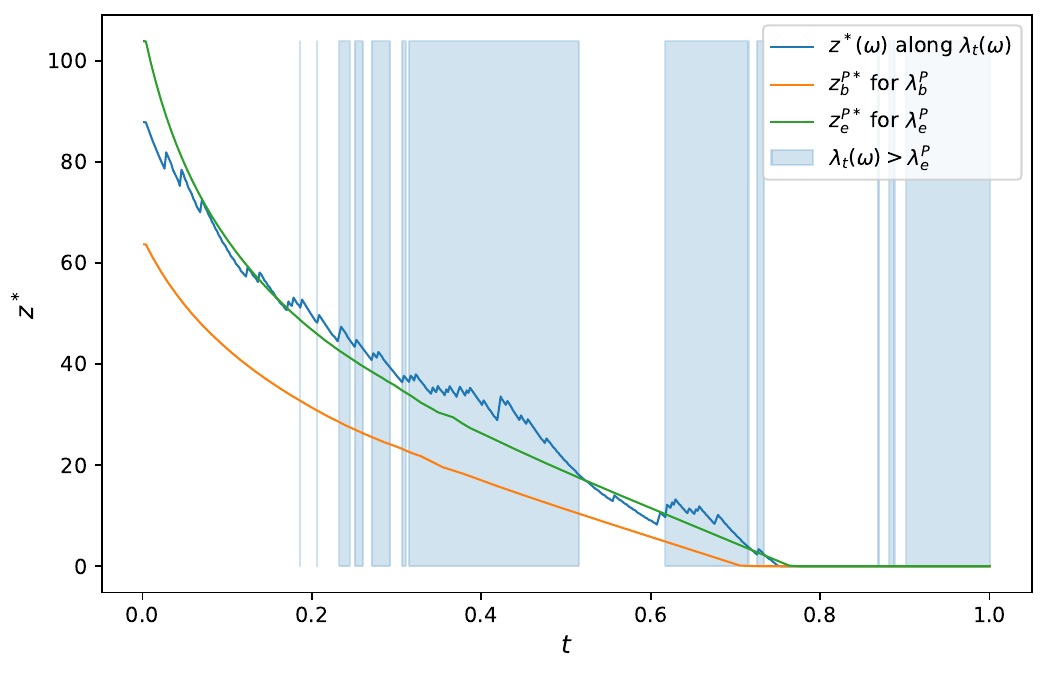}}\\
\subfloat[Intensity path.\label{fig:int traj 2}]{\includegraphics[scale=0.5]{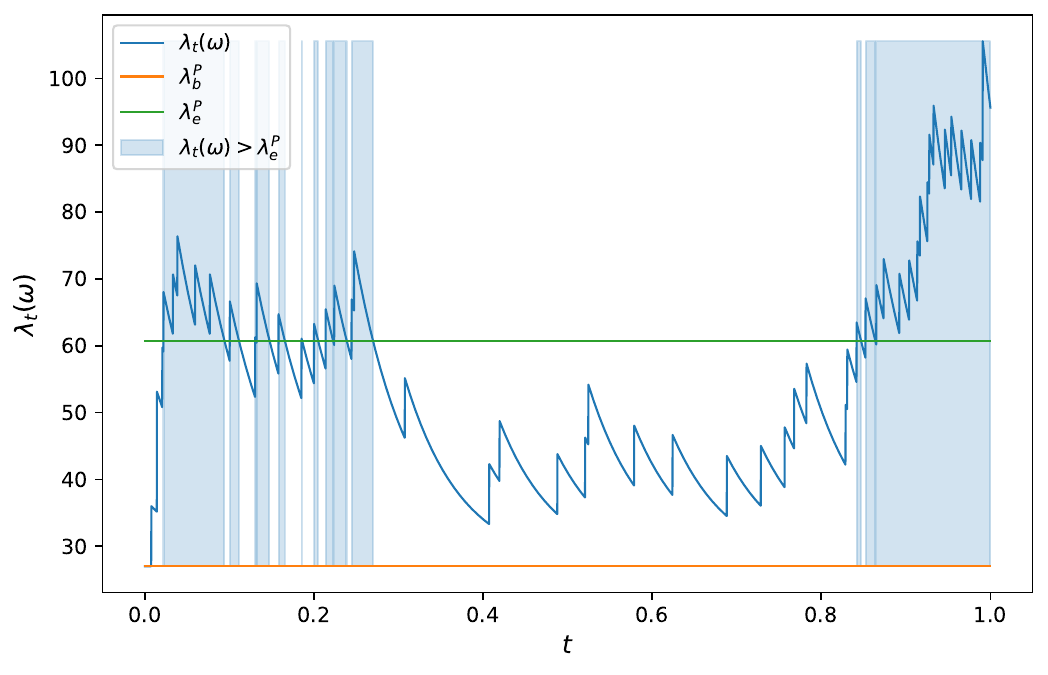}}
    \subfloat[Optimal strategies along the simulated intensity path.\label{fig:control traj 2}]{\includegraphics[scale=0.5]{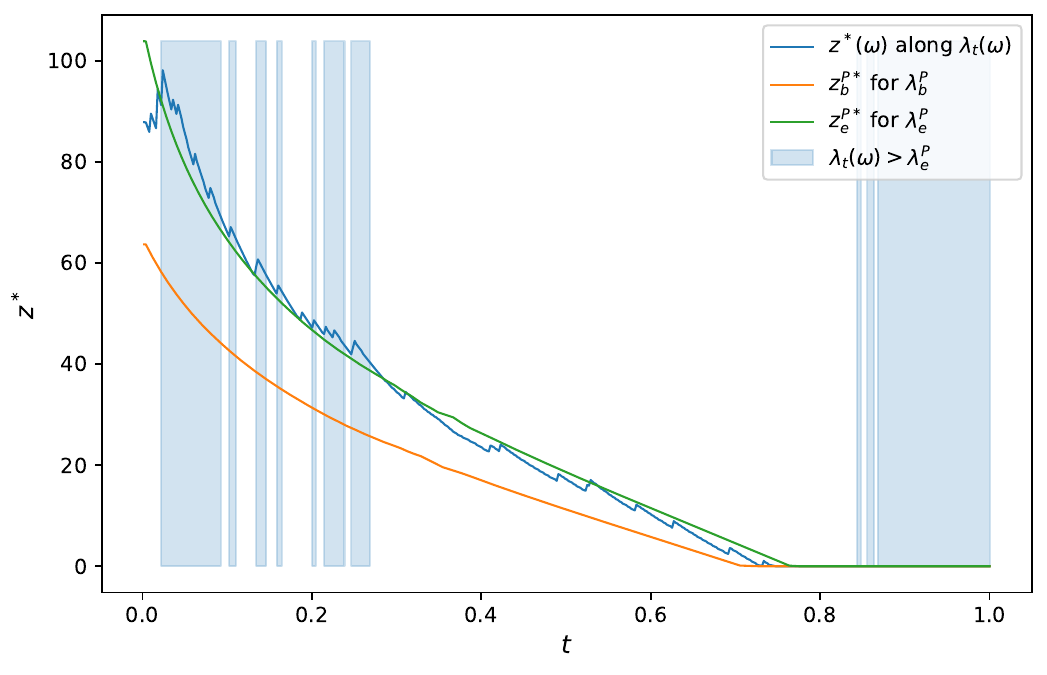}}
    \caption{Optimal strategies along simulated intensity paths.} 
\end{figure*}

\section{Implications for Cyber-insurance}
\label{InsurancePart}

In this section, we connect our dynamic model for optimal cybersecurity investment to actuarial aspects of cyber-insurance. From an insurance viewpoint, the cybersecurity investment policy acts as a form of self-insurance (see \cite{bohme2010modeling}): by investing in cybersecurity, the decision maker changes the distribution of aggregate losses over the planning horizon. Cybersecurity investment is therefore a form of prevention, which represents a crucial aspect of cyber-insurance, as discussed in Section \ref{sec:prevention} below. In turn, prevention measures such as cybersecurity investment have an impact on the determination of cyber-insurance premia, as shown in Section \ref{sec:premia} below.

\subsection{The Role of Prevention in Cyber-insurance}\label{sec:prevention}

Cyber-insurance contracts increasingly go beyond the pure compensation of financial losses and often include prevention requirements and bundled cyber-risk mitigation services. In this sense, cyber-insurance combines risk transfer with risk mitigation services and incentives, promoting best practices in cyber-risk management (see \cite{lopez2025cyber} for a recent analysis).

In our model, prevention is represented by a cybersecurity investment policy $z\in\mathcal{Z}$ and has a precise actuarial meaning: it changes the aggregate losses $L_T^z$ over the planning horizon by reducing the breach probability at attack times from the baseline $v$ to $S(H_{\tau_i}^z, v)$. 
Note that, by lowering the breach probability at attack times, in our model prevention does not only reduce the expected loss, but also reduces loss variability.
As a result, any underwriting or pricing rule that depends on the level of risk should, in principle, depend on the prevention policy (or on observable proxies of it, as discussed below). 
Our framework allows quantifying the effect of optimal prevention, which is an essential input for the insurer in order to determine cybersecurity requirements, incentive schemes, and premium differentiation rules across policyholders (see Section \ref{sec:premia}).

A key practical issue is that prevention is only partially observable, creating asymmetric information and moral hazard. Two implementation channels discussed in the literature are consistent with our framework. 
First, prevention may be delegated to the insurer or to an IT security provider in partnership with the insurer, which is particularly relevant for small firms or policyholders with limited internal expertise. This scheme makes prevention partly contractible and allows joint pricing of risk transfer together with mitigation services (see \cite{zeller2023risk}).
Second, when prevention is handled internally by the entity (as a self-protection measure), insurers may rely on incentive-based contracting. One approach is to condition the premium on minimum protection measures, combined with audits and penalties if requirements are not met (see \cite{hofmann2007internalizing, lelarge2009economic}). 
Another approach is to rely on principal–agent incentive schemes, as considered in \cite{brachetta2025optimal, mastrolia2025agency}, linking part of the compensation to measurable reduction in risk or losses. 
In both cases, the insurer needs quantitative tools to assess the impact of prevention measures, as considered in our study.

\subsection{Implications for Premium Determination} \label{sec:premia}

To quantify the impact of prevention on cyber-insurance premia, we apply the standard deviation principle (see, e.g., \cite{schmidli2017risk}), for an arbitrary cybersecurity investment policy $z\in\mathcal{Z}$:
\begin{equation}\label{eq:insurancepremium}
\pi(L^z_T):= \mathbb{E}[L^z_T] + \theta \,\sigma (L^z_T),
\end{equation}
where $\mathbb{E}[L_T^z]$ is the pure premium and $\theta \sigma(L_T^z)$ is a loading proportional to the standard deviation $\sigma(L^z_T)$ of $L^z_T$, with $\theta$ typically in  $[0.1,0.5]$. This principle is stable with respect to the monetary
unit and captures the impact of prevention on both the expected loss and the dispersion of losses.

We first compare the expected aggregate loss without additional investment, $\mathbb{E}[L_T^0]$, with the expected aggregate loss
under the optimal prevention policy $z^\ast$, denoted $\mathbb{E}[L_T^{z^\ast}]$, corresponding to \eqref{eq:insurancepremium} with $z=z^*$ and $\theta=0$. We compute $\mathbb{E}[L_T^0]$ by
Proposition \ref{prop:expectations}, while $\mathbb{E}[L_T^{z^\ast}]$ is estimated by Monte Carlo over $100{,}000$
samples, where $z^\ast$ is computed via Algorithm \ref{algorithm control intensity trajectory}
starting from $h=0$. Under the standard parameter set, we obtain (in k\$)
\[
    \mathbb{E}[L_T^0]=394.98
    \qquad\text{and}\qquad 
    \mathbb{E}[L_T^{z^\ast}]=141.77,
\]
corresponding to a reduction of approximately $65\%$ in the pure premium under the optimal dynamic prevention strategy.

The variance of the losses $L^z_T$ associated to a given cybersecurity investment policy $z\in\mathcal{Z}$ is given in the following lemma. We denote $\bar{\eta}:=\mathbb{E}[\eta_i]$ and $\sigma^2_{\eta}:=\mathrm{Var}(\eta_i)$, for $i\in\mathbb{N}^*$.

\begin{lemma}\label{lemma:variance}
Suppose that Assumptions \ref{assB} and \ref{assC} hold. Assume furthermore that $\eta_i\in L^2(\mathbb{P})$, for all $i\in\mathbb{N}^*$. Then, for every $z\in\mathcal{Z}$, it holds that
\begin{align*}
& \mathrm{Var}(L^z_T)   \\
&\;= \mathbb{E}\left[\int_0^T\Bigl(\sigma^2_{\eta}\,S(H^z_t,v)+\bar{\eta}^2S(H^z_t,v)\bigl(1-S(H^z_t,v)\bigr)\Bigr)\lambda_t \rmd t\right]    \\
&\quad + \bar{\eta}^2\,\mathrm{Var}\left(\int_0^TS(H^z_t,v)\,\rmd N_t\right).
\end{align*}
\end{lemma}
\begin{proof}
Let $L_T^z$ be defined as in \eqref{actual losses with security}, where $z = (z_t)_{t \in [0,T]}$ is an admissible cybersecurity investment process. 
By the law of total variance (see, e.g., \cite[Problem 34.10]{Billingsley}), we can compute
\begin{equation}\label{eq:variance}\begin{aligned}
    \mathrm{Var}(L_T^z) 
    & = \mathbb{E}\left[\mathrm{Var}\left(\sum_{i=1}^{N_T}  \eta_i B_i^{S(H^z_{\tau_i},v)}\Bigg|\mathcal{F}_T\right)\right] \\
    &\quad 
    + \mathrm{Var}\left(\mathbb{E}\left[\sum_{i=1}^{N_T}  \eta_i B_i^{S(H^z_{\tau_i},v)}\Bigg|\mathcal{F}_T\right]\right)
\end{aligned}\end{equation}
For the first term on the right-hand side in \eqref{eq:variance} it holds that
\begin{align*}
& \mathrm{Var}\left(\sum_{i=1}^{N_T}  \eta_i B_i^{S(H^z_{\tau_i},v)}\Bigg|\mathcal{F}_T\right)
= \sum_{i=1}^{N_T}  \mathrm{Var}\left(\eta_i B_i^{S(H^z_{\tau_i},v)}\Big|\mathcal{F}_T\right)   \\
& = \sum_{i=1}^{N_T}\Bigl(\sigma^2_{\eta}\,S(H^z_{\tau_i},v)+\bar{\eta}^2S(H^z_{\tau_i},v)\bigl(1-S(H^z_{\tau_i},v)\bigr)\Bigr),
\end{align*}
while for the second term we have that
\[
\mathbb{E}\left[\sum_{i=1}^{N_T}  \eta_i B_i^{S(H^z_{\tau_i},v)}\Bigg|\mathcal{F}_T\right]
= \bar{\eta}\sum_{i=1}^{N_T}S(H^z_{\tau_i},v).
\]
Therefore, we can write that
\begin{align*}
\mathrm{Var}(L^z_T)
&= \sigma^2_{\eta}\,\mathbb{E}\left[\sum_{i=1}^{N_T}S(H^z_{\tau_i},v)\right] \\
&\quad + \bar{\eta}^2\,\mathbb{E}\left[\sum_{i=1}^{N_T}S(H^z_{\tau_i},v)\bigl(1-S(H^z_{\tau_i},v)\bigr)\right] \\
&\quad
+ \bar{\eta}^2\,\mathrm{Var}\left(\sum_{i=1}^{N_T}S(H^z_{\tau_i},v)\right).
\end{align*}
Arguing as in the proof of Proposition \ref{prop:expectations} yields the formula for $\mathrm{Var}(L^z_T)$ in the statement of the lemma.
\end{proof}

In the case $z\equiv 0$ (hence $S(0,v)=v$), the formula of Lemma \ref{lemma:variance} simplifies to
\[
   \mathrm{Var}(L_T^0)
   =
   \mathbb{E}[N_T]\big(\sigma^2_{\eta}v + \bar{\eta}^2 v(1-v)\big)
   + \bar{\eta}^2 v^2 \mathrm{Var}(N_T).
\]
This expression can be made fully explicit by using the closed-form expressions for $\mathbb{E}[N_T]$ and $\mathrm{Var}(N_T)$ available for Hawkes processes (see \cite[Section~2.1]{hillairet2025explicit}). 
For the optimal policy $z^\ast$, the quantities appearing in Lemma \ref{lemma:variance} are computed by Monte Carlo over $100{,}000$ samples.

\begin{table}[h!]
    \centering
    \begin{tabular}{c|c|c|c|c}
        $\mathbb{E}[\eta]$ & $\text{Var}(\eta)$ & $\sigma(L_T^0)$ & $\sigma(L_T^{z^*})$ & $1 - [\sigma(L_T^{z^*})/\sigma(L_T^0)]$ \\
       \hline
        10 & 10 & 118.56 & 51.21 & 56.80\% \\
       10 & 50 & 125.04 & 56.58 & 54.75\% \\
       10 & 100 & 132.70 &  62.65 & 52.79\%
    \end{tabular}
    \vspace{1em}
    \caption{Standard deviations.}
    \label{tab: standard deviation}
\end{table}

\begin{table}[h!]
    \centering
    \begin{tabular}{c|c|c|c|c}
        $\mathbb{E}[\eta]$ & $\text{Var}(\eta)$ & $\pi(L_T^0)$ & $\pi (L_T^{z^*})$ & $1 - [\pi (L_T^{z^*})/\pi(L_T^0)]$ \\
       \hline
        10 & 10 & 430.55 & 157.13 & 63.50\% \\
       10 & 50 & 432.5 & 158.74 & 63.32\% \\
       10 & 100 & 434.79 &  160.57 & 63.07\%
    \end{tabular}
        \vspace{1em}
    \caption{Premia (std deviation principle).}
    \label{tab:premium}
\end{table}

Table~\ref{tab: standard deviation} reports the standard deviations of $L_T^0$ and $L_T^{z^\ast}$ for different values of
$\sigma^2_{\eta}$ (with $\bar{\eta}$ fixed), showing a reduction of roughly $53\%$--$57\%$ under optimal prevention.
Table~\ref{tab:premium} reports the corresponding premia under \eqref{eq:insurancepremium} with $\theta=0.3$. Across the
scenarios considered, the premium under optimal prevention is about $63\%$ lower than in the no-investment benchmark. These
results quantify how dynamic optimal cybersecurity investment translates into lower premia.

The values $\pi(L_T^{z^\ast})$ and $\pi(L_T^0)$ provide natural benchmark lower and upper premia corresponding to optimal
prevention and to the absence of additional cybersecurity investment, respectively. 
More generally, when policyholders are classified into prevention
classes using observable proxies (e.g., cybersecurity scores or audits), the model supports premium differentiation across classes: better prevention (corresponding to higher cybersecurity levels) implies lower risk, and therefore lower actuarial premia under standard premium principles. In this sense, the gap between $\pi(L_T^{0})$ and $\pi(L_T^{z^\ast})$ provides a concrete estimate of the premium range induced by prevention under the standard deviation principle.

\section{Conclusions} \label{Conclusions}

In this work, we introduce a dynamic and stochastic extension of the Gordon-Loeb model \cite{gordon2002economics} for optimal cybersecurity investment, incorporating temporally clustered cyberattacks via a Hawkes process. Our modeling framework captures the empirically observed phenomenon of attack bursts, thus offering a more realistic representation of the current cyber-risk environment. We formulate the cybersecurity investment decision problem as a two-dimensional stochastic optimal control problem, maximizing the expected net benefit of cybersecurity investments. We allow for adaptive investment policies that respond in real-time to the arrival of cyberattacks. 

Our numerical results demonstrate that the optimal cybersecurity investment policy consistently outperforms both static benchmarks and Poisson-based models that ignore clustering. In particular, even when Poisson models are calibrated to match the expected attack frequency, they fail to capture the implications of attack clustering on investment timing and magnitude, thus leading to suboptimal investment decisions. Our findings indicate that the optimal dynamic strategy is able to react promptly to attack clusters, offering substantial improvements in expected net benefit in high-risk scenarios.
From an actuarial perspective, we interpret cybersecurity investment as a form of self-insurance, showing that improved prevention lowers both expected losses and loss variability, which in turn reduces cyber-insurance premia under standard premium principles.
Overall, our results underscore the importance of accounting for dynamic and stochastic threat patterns in cybersecurity planning. The proposed framework supports risk managers, insurers and policymakers in designing responsive cybersecurity investment strategies and cyber-insurance pricing schemes that are consistent with the evolving cyber-risk landscape.

Future research directions include the empirical calibration to sector-specific cyber incident data, the consideration of risk-aversion with respect to losses resulting from cyberattacks, and the integration of cyber-insurance as a complementary tool for risk mitigation (see  \cite{awiszus2023modeling,Dou_et_al2020,mazzoccoli2020robustness,miaoui2019enterprise,ougut2011cyber,skeoch2022expanding,SkeochIoannidis24} for some recent studies in this direction). 
Our framework can also be applied from the viewpoint of an insurance firm which provides insurance against losses due to cyberattacks, thus laying the foundations for the development of Cram\'er-Lundberg-type models (see, e.g., \cite{Mikosch}) for cyber-insurance.
Finally, our modeling setup can also be extended to multivariate Hawkes processes (as considered in \cite{Embrechts_et_al2011}, or in the more general versions of \cite{Bielecki_et_al2022,Bielecki_et_al2023})  to differentiate among multiple types of cyberattacks (see \cite{bentley_et_al2020} for a multivariate generalization of the static Gordon-Loeb model).

\section*{Acknowledgements} \label{Acknowledgements} 
The authors are thankful to Alessandro Calvia, Salvatore Federico, Luca Grosset, Eleonora Losiouk, Daniele Marazzina, Andrea Pallavicini, Andrea Perchiazzo, Mario Putti, Edit Rroji for valuable discussions on the topic of this work.
Constructive comments by three anonymous reviewers and an associate editor are gratefully acknowledged.

\section*{Competing Interests}
All authors declare they do not have competing interests.

\section*{Data Availability Statement}
Data availability is not applicable to this article as no new data were created or analysed in this study. If needed, the authors agree to make available the code used for the numerical results illustrated in Section 5.

\section*{Funding} 
This work was supported by the Europlace Institute of Finance; Fondazione Gini; Unione Matematica Italiana and European Union-Next Generation EU [PRIN-2022BE-MMLZ]; and the Chair Stress Test, RISK Management and Financial
Steering, led by the French Ecole Polytechnique and its Foundation and sponsored by BNP Paribas.
The funders had no role in study design, decision to publish, or preparation of the manuscript.
\bibliographystyle{IEEEtran}
\bibliography{refs}

\end{document}